\documentclass[format=sigconf,natbib=true,urlcolor=blue,linkcolor=blue,vvarbb,dvipsnames]{acmart}

\usepackage{natbib}
\usepackage{xcolor}
\usepackage{xspace} \usepackage{todonotes}
\usepackage{bm} \usepackage{mathrsfs}

\newcommand{\CCo}{\mathbb{C}}
\newcommand{\RRo}{\mathbb{R}}
\newcommand{\QQo}{\mathbb{Q}}
\newcommand{\PPo}{\mathbb{P}}

\renewcommand{\aa}{\bm{a}}

\newcommand{\cc}{\bm{c}}
\newcommand{\dd}{\bm{d}}

\newcommand{\ff}{\bm{f}}

\newcommand{\pp}{\bm{p}}

\newcommand{\uu}{\bm{u}}
\newcommand{\bvv}{\bm{v}}
\newcommand{\ww}{\bm{w}}
\newcommand{\xx}{\bm{x}}
\newcommand{\yy}{\bm{y}}
\newcommand{\zz}{\bm{z}}

\newcommand{\CC}{\bm{C}}

\newcommand{\KK}{\bm{K}}

\newcommand{\OO}{\bm{O}}
\newcommand{\PP}{\bm{P}}
\newcommand{\QQ}{\bm{Q}}
\newcommand{\RR}{\bm{R}}

\newcommand{\XX}{\bm{X}}

\newcommand{\afrak}{\mathfrak{a}}

\newcommand{\Afrak}{\mathfrak{A}}

\newcommand{\Efrak}{\mathfrak{E}}

\newcommand{\scal}{\mathcal{s}}

\newcommand{\Dcal}{\mathcal{D}}
\newcommand{\Ecal}{\mathcal{E}}

\newcommand{\Hcal}{\mathcal{H}}

\newcommand{\Kcal}{\mathcal{K}}
\newcommand{\Lcal}{\mathcal{L}}

\newcommand{\Pcal}{\mathcal{P}}

\newcommand{\Rcal}{\mathcal{R}}

\newcommand{\Vcal}{\mathcal{V}}
\newcommand{\Wcal}{\mathcal{W}}

\newcommand{\scrC}{\mathscr{C}}

\newcommand{\scrG}{\mathscr{G}}

\newcommand{\scrP}{\mathscr{P}}

\newcommand{\scrR}{\mathscr{R}}

\newcommand{\scrV}{\mathscr{V}}

\renewcommand{\phi}{\varphi}

 \usepackage{graphicx}
\usepackage{comment}
\usepackage{algorithm}
\usepackage[]{algpseudocode}
\algrenewcommand\algorithmicrequire{\textbf{Input:}}
\algrenewcommand\algorithmicensure{\textbf{Output:}}
\usepackage{enumitem}
\usepackage{pgf,tikz}
\usetikzlibrary{arrows}

\newcommand{\V}{\bm{V}}
\newcommand{\I}{\bm{I}}
\renewcommand{\KK}{\QQ}
\newcommand{\KKbar}{\CC}
\renewcommand{\scal}[1]{\left\langle #1 \right\rangle}
\newcommand{\XXi}[1][n]{x_1,\dotsc,x_{#1}}
\renewcommand{\OO}{\bm{0}}
\newcommand{\ZZo}{\mathbb{Z}}
\newcommand{\Avar}{\textbf{\texttt{A}}}
\newcommand{\Dvar}{\textbf{\texttt{D}}}
\newcommand{\avar}{\afrak}

\renewcommand{\PP}{\PPo}
\newcommand{\PPn}[1][n]{\PP^{#1}}
\newcommand{\Hinf}{\mathcal{H}^{\infty}}

\newcommand{\Pproj}[1]{\mathfrak{p}_{#1}}
\newcommand{\Line}[2]{\Lcal(#1,#2)}
\newcommand{\Grass}{\mathbb{G}}
\newcommand{\Sec}{\textsf{Sec}}
\newcommand{\Tri}{\textsf{Tri}}
\newcommand{\Tang}{\textsf{Tg}}
\newcommand{\SecCotg}{\textsf{CoTg}}

\newcommand{\albe}{\yy}
\newcommand{\Kbarpow}{\KKbar[[x_1-\alpha]]}
\newcommand{\dx}[1]{\partial_{x_1} #1}
\newcommand{\dy}[1]{\partial_{x_2} #1}
\newcommand{\dxx}[1]{\partial^2_{x_1} #1}
\newcommand{\dyy}[1]{\partial^2_{x_2} #1}
\newcommand{\dxy}[1]{\partial^2_{x_1x_2} #1}
\newcommand{\du}[1]{\partial_{u} #1}
\newcommand{\dv}[1]{\partial_{v} #1}

\newcommand{\K}{\Kcal}
\newcommand{\Wo}{\Wcal^{\circ}}
\newcommand{\WoCcur}[1][]{\Wo(\pi_1,\Ccur^{#1})}
\newcommand{\WoCcurde}[1][]{\Wo(\pi_1,\Ccur^{#1}_2)}
\newcommand{\KCcur}[1][]{\K(\pi_1,\Ccur^{#1})}
\newcommand{\KCcurde}[1][]{\K(\pi_1,\Ccur_2^{#1})}
\newcommand{\KCcurR}{\K(\pi_1,\CcurR)}
\newcommand{\KCcurRde}{\K(\pi_1,\CcurRde)}

\newcommand{\KCcurRode}{\K(\pi_1,\CcurRode)}

\newcommand{\Ccur}{\scrC}
\newcommand{\Pset}{\Pcal}

\newcommand{\PCcur}{\overline{\Ccur}}
\newcommand{\CcurR}{\Ccur_{\RR}}
\newcommand{\CcurRde}{\Ccur_{2,\RR}}
\newcommand{\CcurRtre}{\Ccur_{3,\RR}}
\newcommand{\PsetR}{\Pset_{\RR}}
\newcommand{\PsetRde}{\Pset_{2,\RR}}

\newcommand{\CcurRo}{\Ccur_{\RRo}}
\newcommand{\CcurRode}{\Ccur_{2,\RRo}}

\newcommand{\PsetRode}{\Pset_{2,\RRo}}

\newcommand{\Agen}{\Afrak}

\newcommand{\Ggen}{\Efrak}
\newcommand{\Hspace}[2][]{(\mathsf{H_{#2}{#1}})}

\newcommand{\Gg}{\scrG}
\newcommand{\Eg}{\Ecal}
\newcommand{\Vg}{\Vcal}
\newcommand{\Vapp}{\Vg_{\textup{app}}}
\newcommand{\Vcrit}{\Vg_{\K}}

\newcommand{\Ggp}{\Gg_2}
\newcommand{\Vgp}{\Vg_2}
\newcommand{\Egp}{\Eg_2}
\newcommand{\Hiso}{\Hcal}
\newcommand{\Hun}{\Hiso_1}
\newcommand{\Hisot}{\tilde{\Hcal}}
\newcommand{\Hunt}{\tilde{\Hiso}_1}
\newcommand{\Vseq}{\scrV}
\newcommand{\Vseqapp}{\Vseq_{\textup{app}}}

\newcommand{\VseqPparam}{\Vseq_{\Pparam}}

\newcommand{\vtheta}{\vartheta}

\newcommand{\Rparam}{\scrR}
\newcommand{\Cparam}{\scrR}
\newcommand{\Dparam}{\scrR_2}
\newcommand{\Pparam}{\scrP}
\newcommand{\Qparam}{\scrP_2}
\newcommand{\Zparam}{\mathsf{Z}}

\newcommand{\deltaP}{\mu}
\newcommand{\tauP}{\kappa}
\newcommand{\qapp}{q_{\textup{app}}}

\newcommand{\infx}{\preceq_1}

\newcommand{\TopoPlane}{\textsf{Topo2D}\xspace}

\newcommand{\AppSing}{\textsf{ApparentSingularities}\xspace}
\newcommand{\NodeRes}{\textsf{NodeResolution}\xspace}

\newcommand{\IndexConComp}{\textsf{IndConnectComp}\xspace}

\newcommand{\ConnectCurve}{\textsf{ConnectCurve}\xspace}

\newcommand{\SA}{s.a.\xspace}
\newcommand{\SAC}{\SA connected\xspace}
\newcommand{\SACC}{\SA connected component\xspace}
\newcommand{\SACCs}{\SA connected components\xspace}
\newcommand{\ZD}{zero-dimensional\xspace}
\newcommand{\OD}{one-di\-men\-sio\-nal\xspace}
\newcommand{\ZDP}{\ZD parametrization\xspace}
\newcommand{\ZDPs}{\ZD parametrizations\xspace}

\newcommand{\ODP}{\OD pa\-ra\-met\-ri\-za\-tion\xspace}
\newcommand{\ODPs}{\OD parametrizations\xspace}
\newcommand{\NEZO}{non-empty Zariski open\xspace}

\newcommand{\et}{\quad \text{and} \quad}

\newcommand{\Otilde}{\tilde{O}}
\newcommand{\Ahom}{\tilde{A}}

\DeclareMathOperator{\app}{app}

\DeclareMathOperator{\sing}{sing}

\DeclareMathOperator{\jac}{Jac}
\DeclareMathOperator{\reg}{reg}

\DeclareMathOperator{\GL}{GL}
\DeclareMathOperator{\Res}{Res}

\DeclareMathOperator{\sr}{sr}

\title{Algorithm for Connectivity Queries on Real Algebraic Curves}

\author{Md Nazrul Islam}
\orcid{XXX}
\affiliation{\institution{Diebold Nixdorf}
	\city{London}
	\country{Canada}}
\email{mdnazrul.islam@dieboldnixdorf.com} 

\author{Adrien Poteaux}
\orcid{0000-0002-7493-3001}
\affiliation{\institution{Univ. Lille, CNRS, Centrale Lille, UMR 9189 CRIStAL}
  \postcode{F-59000}
  \city{Lille}
  \country{France}
}
\email{adrien.poteaux@univ-lille.fr}

\author{R\'emi Pr\'ebet}
\orcid{0000-0002-3630-6242}
\affiliation{\institution{Sorbonne Universit\'e, \textsc{CNRS}, \textsc{LIP6}}
	\city{F-75005 Paris}
	\postcode{75252}\country{France}}
\email{remi.prebet@lip6.fr}

\theoremstyle{plain}
\newtheorem{theorem}{Theorem}[section]
\newtheorem{lemma}[theorem]{Lemma}
\newtheorem{proposition}[theorem]{Proposition}
\newtheorem*{proposition*}{Proposition}

\newtheorem{corollary}[theorem]{Corollary}

\theoremstyle{definition}
\newtheorem{definition}[theorem]{Definition}

\theoremstyle{remark}

\newcommand{\remi}[1]{\textcolor{black}{#1}\xspace}
\newcommand{\hide}[1]{}
\newcommand{\adrien}[1]{\textcolor{black}{#1}\xspace}
\newcommand{\review}[1]{\textcolor{black}{#1\xspace}}

\renewcommand{\shortauthors}{N. Islam, A. Poteaux, R. Pr\'ebet}

\thanks{We thank I. Bannwarth, M. Safey El Din
  and \'{E}. Schost for sharing some preliminary private notes on this
  problem, which helped to understand the difficulties to lift.
  The last author is supported by the joint \grantsponsor{anr}{Agence
    nationale de la recher\-che (ANR)}{https://anr.fr}
  \grantsponsor{fwf}{Austrian Science Fund FWF}{https://www.fwf.ac.at/en/} grant
  agreement \grantnum{anr}{ANR-FWF ANR-19-CE48-0015} \textsc{ECARP} project.}

\begin{document}

\begin{abstract}
We consider the problem of answering connectivity queries on a real algebraic
curve. The curve is given as the real trace of an algebraic curve, assumed to
be in generic position, and being defined by some rational parametrizations.
The query points are given by a zero-dimensional parametrization.

We design an algorithm which counts the number of connected components of the 
real curve under study, and decides which query point lie in which connected 
component, in time log-linear in $N^6$, where $N$ is the maximum of the degrees 
and coefficient bit-sizes of the polynomials given as input.
This matches the currently best-known bound for computing the topology of real 
plane curves.

The main novelty of this algorithm is the avoidance of the computation of the
complete topology of the curve. 
\end{abstract}

\maketitle
\section{Introduction}
\label{sec:intro}
This work addresses the problem of designing an algorithm for
answering  connectivity queries on real algebraic curves in $\RR^n$,
defined as real traces of algebraic curves of $\KKbar^n$.
\vspace*{-0.3em}
\paragraph*{Motivation and problem statement} 
Consider a real field $\KK$\footnote{Note that in Sections
  \ref{sec:generic} and \ref{sec:apparent}, $\KK$ can be any
  arbitrary field of characteristic 0.}, its real closure $\RR$ and its
algebraic closure $\KKbar$. For $n \geq 1$, let $\XX=(\XXi)$ be a
sequence of indeterminates, and denote $\KK[\XX]$ and $\KKbar[\XX]$
the rings of multivariate polynomials in the $x_i$'s, with
coefficients in resp. $\KK$ and $\KKbar$. We define an \emph{algebraic set}
$\Ccur \subset \KKbar^n$ as the set of common zeros $\V(f_1,\dotsc, f_p)$ of a
sequence of polynomials $(f_1,\dotsc, f_p)\subset \KKbar[\XX]$.
$\I(\Ccur) \subset \KKbar[\XX]$ is the radical of the ideal
$\scal{f_1\dotsc,f_p}$ generated by the $f_j$'s, that is the
  \emph{ideal of definition of $\Ccur$}. The \emph{function ring}
$\KKbar[\Ccur]$ of polynomial functions defined on $\Ccur$ is
$\KKbar[\XX]/\I(\Ccur)$. If
$\I(\Ccur) \subset \KK[\XX]$, we also denote $\KK[\Ccur]$ by
$\KK[\XX]/\I(\Ccur)$. Finally, $\Ccur$ is an
\emph{algebraic curve} if $\I(\Ccur)$ is equidimensional of
dimension 1, and \emph{plane} if contained in
some plane of $\KKbar^n$.

In this document, $\Ccur$ is an algebraic curve such that
$\I(\Ccur) \subset \KK[\XX]$. Given a generating system $\ff$ of
$\I(\Ccur)$, $\jac(\ff)$ is the Jacobian matrix of $\ff$,
$\sing(\Ccur)$ the set of singular points of $\Ccur$ (i.e. the
points where $\jac(\ff)$ has rank less than $n-2$; it is a finite subset
of $\Ccur$) and $\reg(\Ccur) = \Ccur - \sing(\Ccur)$. For all
$\xx \in \reg(\Ccur)$, $T_{\xx}\Ccur$ is the right-kernel of
$\jac(\ff)$: it is the tangent line of $\Ccur$ at $\xx$.  For
$1\leq i \leq n$ we let $\pi_i:\KKbar^n \to \KKbar^i$ be the
canonical projection on the first $i$ variables. If
$\Ccur_2\subset \KKbar^2$ is the Zariski closure of
$\pi_2(\Ccur)$, the set of apparent singularities of $\Ccur_2$ is
$\app(\Ccur_2) = \sing(\Ccur_2) - \pi_2(\sing(\Ccur))$. These
are the singularities introduced by $\pi_2$. A singular point of
$\Ccur_2$ is called a node if it is an ordinary double point (see
\cite[\S3.1]{Ka2008}).  We refer to \cite{Sh2013} for definitions and
propositions about algebraic sets.

For any $\phi \in \KKbar[\Ccur]$, we denote by $\Wo(\phi,\Ccur)$ the set of 
\emph{critical points of $\phi$ on $\Ccur$}, that is the set of points $\xx \in 
\reg(\Ccur)$ such that $d_{\xx}\phi:T_{\xx}\Ccur\to \KKbar$ is not surjective.
Then we note 
\[
  \K(\phi,\Ccur)=\Wo(\phi,\Ccur)\cup \sing(\Ccur)
\]
the set of \emph{singular points of $\phi$ on $\Ccur$}.

To satisfy some genericity assumptions, we will need to perform some
linear changes of variables. Given $A \in \GL_n(\KKbar)$, for
$f \in \KKbar[\XX]$, $f^A$ will denote the polynomial $f(AX)$. For
$V \subset \KKbar^n$, we denote by $V^A$ the image of $V$ by the map
$\Phi_A:\xx \mapsto A^{-1}\xx$. Thus, for
$\ff = (f_1, \dotsc, f_p)\subset \KKbar[\XX]$ we have
$\V(\ff^A) = \Phi_A(\V(\ff)) = \V(\ff)^A$.

A \emph{semi-algebraic (\SA) set} $S\subset\RR^n$ is the set of solutions of a
finite system of polynomial equations and inequalities with coefficients in
$\RR$. We say that $S$ is \emph{\SAC} if for any $\yy, \yy' \in S$, $\yy$ and
$\yy'$ can be connected by a \emph{\SA path} in $S$, that is an injective
continuous \SA function $\gamma\colon[0,1]\to S$ such that $\gamma(0)=\yy$ and
$\gamma(1)=\yy'$. A \SA set $S$ can be decomposed into finitely many
\emph{\SACCs} which are \SAC \SA sets that are both closed and open in $S$. We
refer to \cite{BPR2016} and \cite{BCR1998} for definitions and propositions
about \SA sets and functions. In this work, the \SA sets in consideration
will mainly be real traces of algebraic sets of $\KKbar^n$ (defined by
polynomials with coefficients in $\RR$). In particular, we
will note e.g. $\CcurR$ and $\CcurRde$, respectively the real traces of
$\Ccur$, $\Ccur_2$. Then, e.g. $\KCcur\cap\RR^n$ and $\KCcurde\cap\RR^2$
will be denoted by $\KCcurR$ and $\KCcurRde$.
\smallskip

In this paper, we address the problem of designing an \textbf{algorithm for 
answering  connectivity queries on real algebraic curves in $\RR^n$}, defined 
as real traces of algebraic curves of $\KKbar^n$. More precisely, given 
representations of an algebraic curve $\Ccur$ and a finite set $\Pset$ of 
points of $\Ccur$, we want to compute a partition of $\Pset$, grouping the 
points lying in the same \SACC of $\CcurR$.

It is a problem of importance in symbolic computation, and more specifically, 
in effective real algebraic geometry. Indeed, using the notion of 
roadmaps, introduced by Canny in \cite{Ca1988,Ca1993}, one can reduce 
connectivity queries in real algebraic sets of arbitrary dimension to the such 
queries on real algebraic curves. 
Moreover, algorithms computing such roadmaps, on input a real algebraic set, 
has 
been continuously improved in a series of recent works 
\cite{BPR1999,SS2017,BR2014,BRSS2014}, making now tractable challenging 
problems in 
applications such as robotics \cite{Ca1988,CSS2020,CPSSW2022,CSS2023}.

We say that $f \in \ZZo[x_1,\dotsc,x_n]$ has magnitude
$(\delta,\tau)$, if the total degree of $f$ is bounded by $\delta$ and
all coefficients have absolute values at most $2^\tau$. This extends
to a sequence of polynomials by bounding all entries in the same
way. Complexity results are expressed with $(\delta,\tau)$ bounding
the magnitude of the polynomials defining $\Ccur$.
\review{Moreover, we ignore logarithmic factors using the soft-Oh notation 
$\Otilde(g)$ for denoting the class $g\log(g)^{O(1)}$.}
\paragraph*{Prior works}
One can reduce our problem to a piecewise linear
approximation sharing the same topology as the curve under study.

Computing the topology of plane algebraic curves in $\RRo^2$
is extensively studied: by subdivision algorithm
\cite{BCGY2008,LMP2008}, variants of Cylindrical Algebraic
Decomposition methods
\cite{BEKS2013,CLPPRT2010,DDRRS2022,Di2009,DRR2014,EKW2007,
  GK1996,KS2015,KS2012,MSW2015,SW2005,DET2007,DET2009}, or also a
hybrid approach such as \cite{AMW2008}. In particular, \cite{KS2015,
  DDRRS2022} obtain the best-known complexity bound in
$\Otilde(\delta^{5}(\delta+\tau))$, by computing quantitative bounds
on (bivariate) real root isolation of the considered polynomials.

The problem in $\RRo^3$ has been less studied. This is done through
various approaches such as computing the topology of the projection
on various planes \cite{AS2005,GLMT2005,CJL2013} or lifting the plane
projection by algebraic considerations \cite{Ka2008,DMR2008,DMR2012}.
Yet, few of these papers give a complexity bound for the
computation of such topology \cite{CJL2013,DMR2012}, and \cite{JC2021}
obtains the best-known complexity in
$\Otilde(\delta^{19}(\delta+\tau))$.
\paragraph*{Main result} 
Under genericity assumptions, we reduce the study of a curve $\Ccur$ in 
$\RRo^n$ to the one of the its image $\CcurRtre$ by the projection $\pi_3$, as 
their real traces generically share the same connectivity properties.
Moreover, by refining the approach developed in \cite{NP2011} (based on
\cite{Ka2008}), we show that \emph{one does not need to compute the
topology of $\CcurRtre$ in order to answer connectivity queries}.
More precisely, under genericity assumptions, that we made explicit 
below, we first compute the topology of $\CcurRde$ i.e. an isotopic graph.
Next, the connectivity of $\CcurRtre$ i.e. a homeomorphic graph, is deduced 
from the topology of $\CcurRde$, adapting results from \cite{Ka2008}.
A geometric outcome is that the topological analysis
needed to be done at some special points of $\CcurRde$, which are
called nodes, can be much simplified when one only needs to answer
connectivity queries. This has a significant impact on the
complexity.

Before providing our complexity result, let us
introduce how our geometric objects are encoded.
For a univariate function $\phi$, $\phi'$ is its derivative.
For a bivariate function $\psi$ in the variables $x_1$ and $x_2$, we let
$\dx{\psi},\dy{\psi}$, $\dxx{\psi}$, $\dyy{\psi}$ and $\dxy{\psi}$ be
respectively the simple and double derivative with respect to the index 
variable(s).

To encode finite sets of points with algebraic coordinates over a field 
$\KK$, we use \emph{\ZDPs} $\Pparam = 
(\Omega, \lambda)$ such that
\begin{itemize}
\item $\Omega = (\omega, \rho_1, \ldots, \rho_n)\subset\KK[u]$
where $u$ is a new variable, $\omega$ is a monic square-free polynomial, and
$\deg(\rho_i) <  \deg(\omega)$;
\item $\lambda$ is a linear form $\lambda_1 x_1 + \cdots + \lambda_n x_n$ in 
$\KK[x_1, \ldots, x_n]$ such that
$\lambda_1 \rho_1 + \cdots + \lambda_n \rho_n = u \du{\omega} \mod \omega$.
\end{itemize}
We define the {\em degree} of such a parametrization $\Pparam$ as the degree of 
the polynomial $\omega$, and we say that it encodes the finite set:
\[
  \Zparam(\Pparam) = \left \{\left (\rho_1/{\du{\omega}}, 
  \ldots, 
  \rho_n/{\du{\omega}}
  \right)(\vartheta) \in \KKbar^n 
  \mid \omega(\vartheta) = 0\right \}.
\]

Similarly, we encode algebraic curves with \ODPs over $\KK$, i.e.
$\Cparam = (\Omega, (\lambda, \mu))$ with: 
\begin{itemize}
\item $\Omega = (\omega, \rho_1, \ldots, \rho_n)\subset\KK[u,v]$ with $u$ and 
$v$ new variables, $\omega$ square-free and monic in $u$ 
and $v$, and  $\deg(\rho_i) < \deg(\omega)$;
\item $\lambda=\lambda_1 x_1 + \cdots + \lambda_n x_n$,
  $\mu=\mu_1 x_1 + \cdots + \mu_n x_n$ are linear forms
\end{itemize}
such that $\left\{
  \begin{array}{ll}
    \lambda_1 \rho_1 + \cdots + \lambda_n \rho_n = u \dv{\omega} \mod 
\omega\\\mu_1 \rho_1 + \cdots + \mu_n \rho_n = v \dv{\omega} \mod \omega
  \end{array}
\right.$\\[0.1em]Such a data-structure encodes the algebraic curve $\Zparam(\Cparam)$, defined 
as the Zariski closure of the following locally closed set of $\KKbar^n$:
\[
 \left \{\left (\rho_1/{\dv{\omega}}, 
\ldots, 
\rho_n/{\dv{\omega}}
\right)(\vartheta, \eta) \in \KKbar^n
\,\middle|\,
\omega(\vartheta, \eta) = 0, \:\dv{\omega}(\vartheta, 
\eta)\neq 0\right \}.
\]
We define the {\em degree} of such a parametrization $\Cparam$ as the degree of 
$\omega$, which coincides with the degree of $\Zparam(\Cparam)$. Note that such 
a parametrization $\Cparam$ of degree $\delta$ involves $O(n\delta^2)$ 
coefficients.

We now give our aforementioned genericity properties, which can be seen as
a generalization of the ones in \cite{Ka2008}:
let $\Ccur\subset \KKbar^{n}$ be an algebraic curve and 
$\Pset\subset\reg(\Ccur)$ finite. $(\Ccur,\Pset)$ satisfies $\Hspace{}$ if:
\begin{enumerate}[label=${\Hspace{\arabic*}}$]
 \item\label{hyp:noeth} for $1\leq i \leq n$, $\KK[\Ccur]$ is integral over
$\KK[\Ccur_i]$, where \\$\Ccur_i = \pi_i(\Ccur)$ is an algebraic curve;

 \item\label{hyp:projtang} for all $\xx \in  \reg(\Ccur)$, 
$\pi_2(T_{\xx}\Ccur)$ is a tangent line to $\Ccur_2$ at $\pi_2(\xx)$;

 \item\label{hyp:bij} the restriction of $\pi_3$ to $\Ccur$ is 
injective;

\item\label{hyp:app} if $\yy \in \app(\Ccur_2)$ then
 \begin{itemize}
 \item[${\Hspace['\,]{\arabic{enumi}}}$]\label{hyp:double} 
$\pi_2^{-1}(\yy)\cap\Ccur$ has cardinality 2;
 \item[${\Hspace['']{\arabic{enumi}}}$]\label{hyp:node} $\yy$ is a node of 
$\Ccur_2$;
 \end{itemize}
 
 \item\label{hyp:inj} $\KCcurde \cup \pi_2(\Pset)$ is finite and 
$\pi_1$ is injective on it;

 \item\label{hyp:nosingsecant} $\pi_2^{-1}(\pi_2(\xx))\cap \Ccur = \{\xx\}$,
 for all $\xx \in  \KCcur \cup \Pset$;

\item\label{hyp:odp} there is a \ODP $\Rparam=(\Omega,(x_1,x_2))$
  encoding $\Ccur$, with
  $\Omega=(\omega,x_1,x_2,\rho_3, \dotsc, \rho_n) \subset
  \KK[x_1,x_2]$.
\end{enumerate}
We omit $\Pset$ when the context is clear. Also, $\Hspace{}$ is satisfied
up to a generic linear change of coordinates over $\Ccur$ (see Section
\ref{sec:generic}).

\begin{theorem}\label{thm:mainresult}
 Let $\Cparam \subset \ZZo[x_1,x_2]$ be a \ODP encoding an algebraic 
curve $\Ccur \subset \CCo^n$ satisfying $\Hspace{}$ and $\Pparam \subset 
\ZZo[x_1]$ a \ZDP encoding a finite subset of $\Ccur$. 
Let $(\delta,\tau)$ and $(\deltaP,\tauP)$ the magnitudes of $\Cparam$ and 
$\Pparam$, respectively. 

There exists an algorithm which, on input $\Cparam$ and $\Pparam$, 
computes a partition of the points of $\Zparam(\Pparam)\cap\RRo^n$ lying in the 
same \SACC of $\Ccur\cap\RRo^n$, using 
\[
 \Otilde(\delta^6 + \deltaP^6 + \delta^5\tau + \deltaP^5\tauP)
\]
bit operations.
\end{theorem}
This is to be compared with the best complexity 
$\Otilde(\delta^{19}(\delta+\tau))$ known to analyze the topology of space 
curve.
\review{Note that the dependency on $n$ in the complexity bound is ``hidden'' 
within the potential degrees of the parametrizations and the corresponding 
algebraic sets.
Indeed, according to B\'ezout's bound, an algebraic set, defined by 
polynomials, of degree at most $D$, can have degree at most $D^n$.} 

\paragraph*{Structure of the paper}
After some preliminary results we prove that up to a generic change coordinate,
assumption $\Hspace{}$ holds. Then, under these assumptions, we describe two 
steps
of our algorithm that is identifying the finitely many points of the curve where
there is connectivity ambiguity and resolving these ambiguities. Finally, we
describe the main algorithm together with complexity bounds.

\section{Curves in generic position}
\label{sec:generic}
We now prove that $\Hspace{}$ holds for an algebraic curve in 
generic position $\Ccur$ that is, there is an open dense subset $\Agen$ of 
$\GL_n(\KKbar)$ such that for any $A\in\Agen$ the \adrien{sheared} curve 
$\Ccur^{A}$ 
satisfies 
$\Hspace{}$. 

\subsection{Generic projections of affine curves}
The results below are well known in the case of smooth projective curves (see
e.g. \cite[IV. Thm 3.10]{Ha1977} or \cite[\S7B.]{Mu1995} for $\KKbar=\CCo$),
and have been generalized subsequently in e.g. \cite{HR1979,KKT2008}. A 
version for complex singular affine space curves is proved in \cite[Prop 
5.2]{FGT2009} under regularity assumptions.
 We present here a generalization of \cite[Prop
5.2]{FGT2009} for any singular (affine) algebraic curve, following the proof 
and using more general objects and results from the literature.

Let $n \geq 3$, $\Ccur\subset\KKbar^n$ an affine
algebraic curve and $\Pset \subset \Ccur$ a finite subset. Denote $\PPn$
the projective space $\PPn(\KKbar)$, of dimension $n$ over $\KKbar$, and write
$[\xx_0\colon\cdots\colon\xx_n]$ its elements. Let 
$\Hinf=\{[\xx_0\colon\cdots\colon\xx_n]\in
\PPn \mid \xx_0=0\}$ be the hyperplane at infinity with respect to the affine 
open chart given by $\PPn-\Hinf$ (see e.g. \cite[I.2]{Ha1977})
We finally let $\PCcur$ be the projective closure of $\Ccur$ in $\PPn$.

We denote by $\Grass(1,n)=G(2,n+1)$ the Grassmanian of lines in $\PPn$, and,
for $\xx\neq\yy$ in $\PPn$, by $\Line{\xx}{\yy} \in
\Grass(1,n)$ the line containing $\xx$ and $\yy$.
For distinct points $\xx,\yy$ of $\PCcur$, the line $s=\Line{\xx}{\yy}$ will
be called the \emph{secant line of $\PCcur$ determined by $\xx$ and $\yy$}.
When $s$ intersects $\PCcur$ in a third point, distinct from $\xx$ and $\yy$,
we call it a \emph{trisecant line of $\PCcur$}.
If there are distinct $\xx',\yy' \in s \cap \reg(\PCcur)$ such that
$T_{\xx'}\PCcur$ and $T_{\yy'}\PCcur$ are coplanar, then it will be called a 
\emph{secant line with coplanar tangents of $\PCcur$}.
Then, we define $\Sec(\PCcur), \Tri(\PCcur)$ and $\SecCotg(\PCcur)$ as the sets 
of points in $\PPn$  that lie on respectively a secant, trisecant and secant 
with coplanar tangents of $\PCcur$. 
Finally, we denote by $\Tang(\PCcur)$ the set of points in $\PPn$ that lie on 
the tangent line $T_{\xx}\PCcur$ for some $\xx \in \reg(\PCcur)$.

\begin{lemma}\label{lem:dimsec}
The sets $\Sec(\PCcur)$ and $\Tang(\PCcur)$ are algebraic sets of dimension 
$\leq 3$ and $\leq 2$, respectively.
If, in addition, $\PCcur$ is not a plane curve, then $\Tri(\PCcur)$ and 
$\SecCotg(\PCcur)$ are algebraic sets of dimension $\leq 2$.
Finally, none of these sets contains $\Hinf$.
\end{lemma}
\begin{proof}
Let $\PCcur_1,\dotsc,\PCcur_m$ the irreducible components of $\PCcur$,
$i,j \in \{1,\dotsc,m\}$, possibly equal, and $\Sigma_{i,j}
\subset \Grass(1,n)$ the Zariski closure of the image of $\PCcur_i \times
\PCcur_j - \{(\yy,\yy) \mid \yy \in \PCcur_i \cap \PCcur_j\}$ through the
map $(\yy,\zz) \mapsto \Line{\yy}{\zz}$. As the image of a
Cartesian product of two irreducible curves, $\Sigma_{i,j}$ is an irreducible
algebraic set. Such a secant being uniquely determined by fixing
two points in $\PCcur_i$ and $\PCcur_j$, $\Sigma_{i,j}$ has dimension
$\leq 2$ by \cite[Thm 1.25]{Sh2013}. Then, if $\Sigma = \bigcup_{i,j}
\Sigma_{i,j}$ is the secant variety
of $\PCcur$, it has dimension $\leq 2$ and contains the secant lines
in $\Grass(1,n)$. As elements of $\Grass(1,n)$ \adrien{are}
algebraic sets of dimension 1, $\Sec(\PCcur)$ has Zariski
closure of dimension $\leq 3$.

Consider now, the subset $\Gamma_i \subset \PPn \times \PCcur_i$, consisting 
of  points $(\uu,\yy)$ such that $\yy \in \reg(\PCcur)$ and $\uu \in 
T_{\yy}\PCcur$,
and consider 
the projections $\phi_i\colon\Gamma_i \to \PPn$ and 
$\psi_i \colon\Gamma_i \to \PCcur_i$. For all $\yy$ in the Zariski open subset 
$\reg(\PCcur)\cap \PCcur_i$ of $\Ccur_i$, $\psi_i^{-1}(\yy)$ is exactly 
$T_{\yy}\PCcur$, which has dimension 1.
Hence, by \cite[Thm 1.25]{Sh2013}, $\phi_i(\Gamma_i)$ has Zariski closure of 
dimension $\leq 2$.
Since $\Tang(\PCcur) = \cup_i \phi_i(\Gamma_i)$, we are done.

Assume now, that $\PCcur$ is not a plane curve then, by \cite[Thm
2]{KKT2008}, the set of trisecant lines of $\PCcur$ is a subset of 
$\Grass(1,n)$ whose Zariski closure has dimension $\leq 1$. Then, as seen
above, $\Tri(\PCcur)$ has Zariski closure of dimension $\leq 2$.
Now, let $M_{i,j}$ be the subset of $\Sigma_{i,j}$ consisting of secant lines 
intersecting $\PCcur$ at points whose tangents are all contained in the same
plane. We are going to prove that the Zariski closure of $M_{i,j}$ has 
dimension $\leq 1$. Together with the dimension bound on $\Tri(\PCcur)$,
this will bound the dimension of $\SecCotg(\PCcur)$.
Suppose first that $\PCcur_i$ and $\PCcur_j$ are not coplanar components. Then, 
there is $\yy \in \PCcur_i - \sing(\PCcur)$ such that $l = T_{\yy}\PCcur$ 
and $\PCcur_j$ are not coplanar. If $\Pproj{l}\colon \PPn \to \PPn[n-2]$ 
denotes the projection of center $l$, then $\Pproj{l}(\PCcur_j)$ 
is not a point.
As $\PCcur_j$ is irreducible, and by \cite[Thm 1.25]{Sh2013}, the Zariski 
closure $\Rcal$ of $\Pproj{l}(\PCcur_j)$ is an irreducible algebraic 
subset of $\PPn[n-2]$ of dimension 1.
Hence, by \cite[Thm 1.25]{Sh2013} again, there is a finite set $K_1
\subset \PPn[n-2]$ such that for all $\ww \in \Rcal \backslash K_1$,
$\Pproj{l}^{-1}(\ww) \cap \PCcur_j$ is finite.
Besides, by Sard's Theorem \cite[Thm 2.27]{Sh2013}, there exists a
finite set $K_2 \subset \PPn[n-2]$ such that $\Rcal \backslash K_2$ does not
contain any critical value of the restriction of $\Pproj{l}$ to $\PCcur_j$.
Then,
for $\ww$ in $\Rcal \backslash [K_1 \cup K_2 \cup
\Pproj{l}(\sing(\PCcur))]$, 
\[
  \Pproj{l}^{-1}(\ww) \cap \PCcur_j = \{\zz_1,\dotsc,\zz_k\}
\]  
with $k\geq 1$, and for all $1\leq i\leq k$,
$\zz_i \in \reg(\PCcur)$ and $\Pproj{l}(T_{\zz_i}\PCcur)$ has dimension 1.
Hence, $\yy$ and $\zz_i$ have no coplanar tangents for all $1\leq i\leq k$.
In particular, the secant line $\Line{\yy}{\zz_1}$ contains two points having 
no coplanar tangents so that $\Line{\yy}{\zz_1} \in
\Sigma_{i,j} - M_{i,j}$ and $M_{i,j} \subsetneq \Sigma_{i,j}$.
In conclusion, the Zariski closure of $M_{i,j}$ is a proper algebraic subset, 
and since $\Sigma_{i,j}$ is irreducible, this closure has dimension $\leq 1$.
If $\PCcur_i$ and $\PCcur_j$ are coplanar, $\Sigma_{i,j}$ is the 
Zariski closure of $M_{i,j}$ and one of the following holds.
If $i=j$ and $\PCcur_i$ is a line, then $\Sigma_{i,j}$ is reduced to the line 
associated to $\PCcur_i$ and has dimension 0.
Else, there exists a unique plane $S_{i,j}$ containing
$\PCcur_i$ and $\PCcur_j$, so that any line of $\Sigma_{i,j}$ must be contained 
in 
$S_{i,j}$. In both cases, $\Sigma_{i,j}$, thus the closure of 
$M_{i,j}$, have dimension $\leq 1$.
Then, the Zariski closure of the union $M$ of all $M_{i,j}$ for $i,j \in 
\{1,\dotsc,m\}$, is an algebraic subset of $\Grass(1,n)$ of dimension
$\leq 1$ as requested.

Remark now that a secant with coplanar tangents is either a trisecant, or a
secant intersecting $\PCcur$ in exactly two regular points with coplanar
tangents. Hence, the set of secants with coplanar tangents of $\PCcur$ is
contained in the union of $M$ and the set of trisecant lines of $\PCcur$. By the
previous discussion, it has dimension $\leq 1$, so that the
Zariski closure of $\SecCotg(\PCcur)$ has dimension $\leq 2$.

Since $\PCcur - \Hinf$ can be identified with $\Ccur$, the former is a 
Zariski open subset of $\PCcur$, so that $\PCcur \cap \Hinf$ is 
finite.
In particular, $\Hinf$ contains finitely many secant or tangent lines of 
$\PCcur$ and then, cannot be contained in $\Sec(\PCcur)$ or $\Tang(\PCcur)$. 
Since $\Tri(\PCcur)$ and $\SecCotg(\PCcur)$ are contained in $\Sec(\PCcur)$, 
they cannot contain $\Hinf$ as well.
\end{proof}
\vspace*{-0.3em}
In the following, for $0 \leq r \leq n-1$, we denote by 
$\Grass(r,n-1)=G(r+1,n)$ the set of $r$-dimensional projective linear subspaces 
of $\Hinf$. Recall that using Pl\"ucker embedding (see e.g. \cite[Example 
1.24]{Sh2013}), $\Grass(r,n-1)$ can be embedded in $\PPn[\binom{n}{r+1}\,-1]$ 
as an irreducible algebraic set of dimension $(r+1)(n-r)$.
The next lemma is then a direct consequence of \cite[Thm 1.25]{Sh2013}.
\begin{lemma}\label{lem:grass}
 Let $X \subset \Hinf$ be an algebraic set of dimension $m \leq n-1$.
 Then, for any $i\geq m$ there exists a non-empty Zariski open subset 
$\Ggen_i$ of $\Grass(n-1-i,n-1)$ such that for every $E \in \Ggen_i$, the 
set $E \cap X$ is finite and, if $i>m$, it is empty.
\end{lemma}
Recall that $\Pset$ is a finite set of control points in $\PCcur-\sing(\PCcur)$.
\begin{proposition}\label{prop:projproj}
 If $\PCcur$ is not a plane curve, then for all $1\leq i \leq n-1$, 
there exists a non-empty Zariski open subset $\Ggen_i$ of $\Grass(n-1-i,n-1)$ 
such that for all $E \in 
\Ggen_i$, the following holds. Let $\Pproj{E}:\PCcur \to \PP^i$ be the 
projection with center $E$, then $\Pproj{E}$ is a finite regular map and
 \begin{enumerate}[label=$(\roman*)$]
  \item\label{ite:tangP} for all $\xx \in \Pset$, $\Pproj{E}(T_{\xx}\PCcur)$ is 
a projective line of $\PPn[i]$.
\end{enumerate}
If, in addition, $i\geq2$ then, 
 \begin{enumerate}[label=$(\roman*)$]
 \setcounter{enumi}{1}
 \item\label{ite:tang} item \ref{ite:tangP} holds for any $\xx \in 
\reg{\PCcur}$;
 \item\label{ite:trisec} for any $\xx \in \PCcur$, there exists at most 
one point $\xx' \in \PCcur$, distinct from $\xx$, such that $\Pproj{E}(\xx) 
= \Pproj{E}(\xx')$;
\item\label{ite:sectang} there exists finitely many such couples 
$(\xx,\xx')$, all satisfying $\xx,\xx' \in \reg(\PCcur)-\Pset$ and 
$\Pproj{E}(T_{\xx}\PCcur)\neq \Pproj{E}(T_{\xx'}\PCcur)$;
\item\label{ite:nosec} if $i\geq 3$, there is no such couple.
 \end{enumerate}
\end{proposition}
\begin{proof}
 Fix $1\leq i \leq n-1$ and suppose that $\PCcur$ is not plane.
 As a proper Zariski closed set of 
 $\PCcur$, $X_1 := \Hinf \cap \PCcur$ is finite.
 By Lemma~\ref{lem:grass}, as $i > 0$,
there is a non-empty Zariski open subset $\Ggen_1$ of $\Grass(n-1-i,n-1)$ 
such that for all $E \in \Ggen_1$, $E \cap X_1$ is empty.
Moreover, any $(n-i)$-dimensional space containing $E$ cannot contain an 
irreducible component of $\PCcur$ (it would be a line, 
intersecting $E$ at some point of $E \cap \PCcur = E\cap X_1$, which is empty).
Thus, the projection with center $E \in \Ggen_1$ induces a finite map on 
$\PCcur$, regular by definition.

According to Lemma~\ref{lem:dimsec}, the set of points lying on a tangent or a 
trisecant line of $\PCcur$ is an algebraic set of dimension $\leq 2$. Since 
$\Hinf$ contains finitely many such tangents or trisecants, 
$$X_2 = (\Tang(\PCcur) \cup \Tri(\PCcur)) \cap \Hinf$$ has dimension at most 
$1$. 
By Lemma~\ref{lem:grass}, as $i\geq1$, there exists a non-empty Zariski open 
subset $\Ggen_2$ of $\Grass(n-1-i,n-1)$ such that any $E \in \Ggen_2$ 
intersects finitely many points of $\Tang(\PCcur)\cup\Tri(\PCcur)$. 
Besides, there are finitely many tangents intersecting the finite set 
$\Pset$, so that by Lemma~\ref{lem:grass}, up to intersecting 
$\Ggen_2$ with a non-empty Zariski open subset of $\Grass(n-1-i,n-1)$, one can 
assume that none of these tangents intersect $\Pset$. This proves 
\ref{ite:tangP}.

Assume now $i\geq 2$. By Lemma~\ref{lem:grass}, no $E \in \Ggen_2$ 
intersects points in $\Tang(\PCcur)\cup\Tri(\PCcur)$. In  particular, any 
$(n-i)$-dimensional space containing $E$ cannot contain a tangent nor a 
trisecant, and, as seen above, this means that no tangent, or three 
distinct points, are mapped to one point.
This proves respectively \ref{ite:tang} and 
\ref{ite:trisec}.

Then, by Lemma~\ref{lem:dimsec}, the set $X_3=\Sec(\PCcur) \cap \Hinf$ 
of points
in $\Hinf$, lying on a secant line of $\PCcur$,
is algebraic of dimension
$\leq 2$. By Lemma~\ref{lem:grass} ($i\geq 2$), there is a
non-empty Zariski open subset $\Ggen_3$ of $\Grass(n-1-i,n-1)$ such that any $E
\in \Ggen_3$ contains finitely many points lying on a secant line of $\PCcur$ 
i.e., as before, there are finitely many
couples of points which are mapped to the same point in $\PPn[i]$. Besides,
the set of secants intersecting $\sing(\PCcur) \cup \Pset$ is a
proper algebraic subset of the secant variety of $\PCcur$. Hence, by
Lemma~\ref{lem:grass}, up to intersecting $\Ggen_3$ with a non-empty Zariski
open subset of $\Grass(n-1-i,n-1)$, one can assume that none of these secants
intersect $\sing(\PCcur) \cup \Pset$.
Finally, by Lemma~\ref{lem:grass}, as $\SecCotg(\PCcur) \cap
\Hinf$ has dimension $\leq 1$. As seen above, up to intersecting $\Ggen_3$ with 
a non-empty Zariski open subset of $\Grass(n-1-i,n-1)$, one can assume that 
these secants intersect $\PCcur$ at points with no coplanar tangents, which 
cannot be mapped to the same line. 
\remi{All in all, for any $E \in \Ggen_3$, \ref{ite:sectang} holds.}

By Lemma~\ref{lem:grass}, if moreover $i\geq 3$, no $E \in \Ggen_3$ 
intersects points in $\Sec(\PCcur)$ that is, no two distinct 
points are mapped to the same image. This proves \ref{ite:nosec}. 
Taking $\Ggen_i=\Ggen_1\cap\Ggen_2\cap\Ggen_3$ ends 
the proof. 
\end{proof}

We can now state the affine counterpart of Proposition~\ref{prop:projproj}.
\begin{corollary}\label{cor:projaff}
 There exists a non-empty Zariski open set $\Agen$ of $\GL_n(\KKbar)$ such that 
for 
all 
$A \in \Agen$ and $1 \leq i\leq n$, the following holds:
the restriction of $\pi_i$ to $\Ccur^A$ is a finite morphism, and
 \begin{enumerate}[label=$(\roman*)$]
  \item\label{ite:afftangP} for all $\xx \in \Pset^A$, $\pi_i(T_{\xx}\Ccur^A)$ 
is a line of $\KKbar^i$.
\end{enumerate}
If, in addition, $i\geq2$ then, 
 \begin{enumerate}[label=$(\roman*)$]
 \setcounter{enumi}{1}
  \item\label{ite:afftang} item \ref{ite:afftangP} holds for any $\xx \in 
\reg(\Ccur^{A})$;
 \item\label{ite:affnoninj} the restriction of $\pi_i$ to $\Ccur^A$ is not 
injective at $\xx$ if, and only if, $i=2$ and 
$\pi_2(\xx)\in\app(\Ccur_2^A)$;
 \item\label{ite:affnodes} $\app(\Ccur_2^A)$ contains only nodes, with 
exactly two preimages through $\pi_2$, none of them being in $\Pset^A$;
\end{enumerate}
\end{corollary}
\begin{proof} 
If $\Ccur$ is a plane curve, it is straightforward.
Suppose from now on $n \geq 3$ and $\Ccur$ not plane.
If $i=n$, there is nothing to prove, so let $1\leq i \leq n-1$.
Let $\PCcur$ be the projective closure of $\Ccur$, which is not a plane 
either. 
Let $\Ggen_i$ be the non-empty Zariski open subset of $\Grass(n-1-i, n-1)$ 
given by Proposition~\ref{prop:projproj}.
According to Pl\"ucker embedding, there exists a surjective regular map  
from the set of $i$ linearly independent vectors $\aa_1,\dotsc,\aa_i$ of 
$\KKbar^n$ to the set of $(n-1-i)$-dimensional (projective) linear subspaces of 
$\Hinf$, defined by $x_0=0$ and  $\aa_{j,1}x_1 + \cdots +
\aa_{j,n}x_n=0$ for $1\leq j \leq i$.
Hence, there exists a non-empty Zariski open set $\Agen_i$ of $\GL_n(\KKbar)$ 
of matrices $A$ such that the first $i$ rows of $A^{-1}$ are mapped to some $E 
\in \Ggen_i$, through the above map.
Moreover, for any $A \in \Agen_i$ the following holds. Consider, 
\[
 \tilde{A}= 
\begin{bmatrix}
  1 & \OO\\
  \OO& A
 \end{bmatrix}
,
\]
and for $1\leq j\leq n$, let $\aa_j=(\aa_{j,1},\dotsc,\aa_{j,n})$ be the rows of
$A$. If $L_0 = x_0$ and for $1\leq j\leq i$, $L_j=\aa_{j,1}x_1 + \cdots +
\aa_{j,n}x_n$, then the equations $L_0,\dotsc,L_i$ define a projective linear
subspace $E$ of $\Hinf$, such that $E \in \Ggen_i$ and, by definition
(see e.g. \cite[Example 1.27]{Sh2013}), \vspace*{-0.3em}
\[
 \begin{array}{cccc}
  \Pproj{E}: &\PCcur^{\tilde{A}} &\to &\PPn[i]\\
  &\xx &\mapsto &[\xx_0:\cdots:\xx_i]
 \end{array}.
\]
Therefore, the restriction of $\Pproj{E}$ to the affine chart $\PPn - \Hinf$ 
can be identified with the restriction of $\pi_i$ to $\Ccur^{A}$.
According to Proposition~\ref{prop:projproj}, the restriction of $\pi_i$ to 
$\Ccur^{A}$ is a finite morphism satisfying item \ref{ite:afftangP}.
Assume now that $i\geq2$ then, assertion \ref{ite:afftang} is a direct 
consequence of item \ref{ite:tang} of Proposition~\ref{prop:projproj}.

Besides, let $\xx\in\Ccur^A$ such that there is $\xx'\in\Ccur^A$ satisfying 
$\xx'\neq\xx$ and $\pi_i(\xx)=\pi_i(\xx')$. 
Then, by 
Proposition~\ref{prop:projproj},  \ref{ite:trisec} to \ref{ite:nosec}, $\xx'$ 
is unique, both 
$\xx,\xx'\notin\sing(\Ccur^A)\cup\Pset^A$, and necessarily $i=2$.
Moreover, $T_{\xx}\Ccur^A$ and $T_{\xx'}\Ccur^A$ map to distinct lines of 
$\KKbar^2$, crossing at $\pi_2(\xx)$: it is a node. 
Hence, $\xx \in \app(\Ccur_2^A)$ and $\pi_2(\xx)$ is a node, with exactly two 
preimages, none of them being in $\Pset^A$.
Conversely from 
Proposition~\ref{prop:projproj}, \ref{ite:tang}, all points of 
$\app(\Ccur_2^A)$ have at least 
two preimages in $\Ccur^A$. This proves \ref{ite:affnoninj} and 
\ref{ite:affnodes}. Taking $\Agen = \bigcap_{i=1}^{n-1}\Agen_i$ concludes.
\end{proof}

\subsection{Recovering $\Hspace{}$}

\begin{proposition}
 Let $\Ccur \subset \KKbar^n$ be an algebraic curve and a finite subset $\Pset 
\subset \reg(\Ccur)$.
 There exists a non-empty Zariski open set $\Agen \subset \GL_n(\KKbar)$ such
 that, for any $A \in \Agen$,\: $(\Ccur^A,\Pset^A)$ satisfies $\Hspace{}$.
\end{proposition}
\begin{proof}
  Let $\Agen_1 \subset \GL_n(\KKbar)$ be the non-empty Zariski open subset
  defined in Corollary~\ref{cor:projaff}
  and let $A\in \Agen_1$. For all $1\leq i \leq n$, the restriction of $\pi_i$
  to $\Ccur^{A}$ is a finite morphism, so that $\Ccur^{A}_i =
  \pi_i(\Ccur^{A})$ is an algebraic curve. Since $\KKbar$
  is integral over $\KK$, the extension $\KK[\Ccur^{A}_i] {\hookrightarrow}
  \KK[\Ccur^{A}]$ is integral as well: \ref{hyp:noeth} is satisfied.
  Applying Corollary~\ref{cor:projaff}, for $i=3$ and $i=2$ shows that
  the curve $\Ccur^{A}$ satisfies respectively \ref{hyp:bij} on the
  one hand and \ref{hyp:projtang} and \ref{hyp:app} on the other.

Let $\Avar=(\avar_{i,j})_{1\leq i, j\leq n}$ and $t$ be new indeterminates, 
the former ones standing for the entries of a square matrix of size $n\times n$.
Since $\Agen_1$ is non-empty and Zariski open, there exists a non-zero 
polynomial $F \in \KKbar[\Avar]$, such that $A \in \Agen_1$ if $F(A)\neq 0$.
Besides, according to \cite[\S 4.2]{BK2002} (or \cite[\S 3.2]{GK1996}), there 
exists a non-zero polynomial $G \in \KKbar[\Avar,t]$ such that, if $F(A)\neq 0$ 
and $G(A,b)\neq 0$ then, for
\[
 B= {\small\arraycolsep=0.5\arraycolsep\begin{bmatrix}
     1 & b & \OO\\
     0 & 1 & \OO\\
     \OO & \OO & \scalebox{0.9}{$I_{n-2}$}\\
    \end{bmatrix}},
\]
the curve $\Ccur_2^{BA}$ is a plane curve in generic position in the sense of 
\cite[\S 4.2]{BK2002} and \cite[Def 3.3]{Ka2008}.
In particular, $\pi_1$ maps no tangent line of any singular point of $\Ccur_2$ 
to a point and its restriction of $\pi_1$ to 
the finite set
$\WoCcurde[BA]$ is injective. Let $\Pset_2=\pi_2(\Pset)$. As 
$\Pset_2\cup\sing(\Ccur_2)$ is finite, we can assume that $\pi_1$ is injective 
on $\Pset_2^{BA} \cup \sing(\Ccur_2^{BA})$ as well.
But, for any $\xx\in\WoCcurde[BA]$, $\pi_1(\xx)$ is a point, so that $\xx$ is 
neither in $\sing(\Ccur_2^{BA})$ nor $\Pset_2^{BA}$, by genericity of 
$\Ccur_2^{BA}$ and item \ref{ite:afftangP} of Corollary~\ref{cor:projaff} 
respectively.
Then, let $b \in \KKbar$ such that $G(\Avar,b)$ is not zero and let 
$B$ be as above. The subset $\Agen_2 \subset \GL_n(\KKbar)$ of elements of the 
form $BA'$ where $F(A')G(A',b)\neq 0$ is a \NEZO subset. Moreover, for any 
$A\in\Agen_2$, $\Ccur^A$ satisfies \ref{hyp:inj}.

Take $A \in \Agen_1\cap\Agen_2$ and let $\xx \in \KCcur[A]\cup\Pset^A$ and
$\yy=\pi_2(\xx)$. Suppose there is $\xx' \in \Ccur^A$ such that $\xx'\neq\xx$
and $\pi_2(\xx')=\yy$. By \ref{ite:affnoninj}, $\xx \in \WoCcur[A]$
and $\yy$ is a node in $\app(\Ccur_2^A)$, with vertical tangent line
$\pi_2(T_{\xx}\Ccur^A)$: this is impossible by above ($A
\in \Agen_2$, so that  $\Ccur_2^A$ is in generic position). Therefore, 
$\Ccur^A$ satisfies \ref{hyp:nosingsecant}.

We proceed similarly for \ref{hyp:odp}. Let $A \in \Agen_1$. By
\ref{hyp:noeth}, $\Ccur^{A}$ is in Noether position (for $\pi_1$). Let
$\Dvar=(\mathfrak{d}_3,\dotsc,\mathfrak{d}_n)$ be new variables. By
\cite[Cor 3.4 \& 3.5]{DL2008}, there is $H \in
\KKbar[\Avar,\Dvar]$ non-zero such that, if $F(A)\neq 0$ and $H(A,\dd)\neq 0$, 
then 
the following holds: if $\mu_{\dd} = x_2 + \dd_3x_3 + \cdots + \dd_n x_n$
is a linear form, then there is $\Rparam =(\omega,\rho_1,\dotsc,\rho_n)
\subset \KK[x_1,v]$ such that $(\Rparam,x_1,\mu_{\dd})$ is a \ODP encoding
$\Ccur^{A}$. Let $\dd \in \KKbar^{n-1}$ such that $H(\Avar,\dd)$ is not
zero and
\[
   C = {\small\arraycolsep=0.5\arraycolsep\begin{bmatrix}
     1 & \OO & \OO\\
     0 & 1 & \dd\\
     \OO & \OO  & \scalebox{0.9}{$I_{n-2}$}\\
    \end{bmatrix}}.
\]
The subset $\Agen_3 \subset \GL_n(\KKbar)$ of elements, of the form $CA'$, 
where $F(A')$ and $H(A',\cc)$ are both not zero, is a \NEZO subset where 
$\Ccur^{A}$ satisfies \ref{hyp:odp}.

Finally, 
for $A\in \Agen:=\Agen_1 \cap \Agen_2
\cap \Agen_3$, $\Ccur^{A}$
satisfies $\Hspace{}$.
\end{proof}

\section{Detect apparent singularities}
\label{sec:apparent}
We generalize the criterion of \cite{Ka2008} used to identify apparent
singularities in plane projection of space curve. 
We keep notations given in Section \ref{sec:intro}, and \textbf{assume for the 
rest of the document that $(\Ccur,\Pset)$ satisfies $\Hspace{}$}.
We start by an adapted version of \cite[Lemma 4.1]{Ka2008} (the
equivalence relation modulo $\I(\Ccur)$ is denoted $\equiv$).
\begin{lemma}\label{lem:powersol}
  Let $(\alpha,\beta)$ be a node of $\Ccur_2$.
  There are exactly two power-series $y_1,y_2 \in \Kbarpow$
  such that for $i=1,2$, if $z_i= 
  \frac{\rho_3(x_1,y_i)}{\dy{\omega}(x_1,y_i)}$ then:
  \begin{enumerate}\vspace*{-0.2em}
  \item\label{ass:power1} $\omega(x_1,y_i)\equiv 0$ and 
    $y_i(\alpha)=\beta$ but $y_1'(\alpha)\neq y_2'(\alpha)$;\\[-0.8em]
  \item\label{ass:power2} $h(x_1,y_i,z_i)\equiv 0$ for any 
    $h\in\I(\Ccur) \cap \KK[x_1,x_2,x_3]$\\
    and $z_i \in \Kbarpow$.
  \end{enumerate}
\end{lemma}
\begin{proof}
  According to \ref{hyp:inj} and \ref{hyp:odp}, 
  $\Ccur_2$ is in generic position in the sense of \cite[Def 3.1]{GK1996}.
As $(\alpha,\beta)$ is a node 
of $\Ccur_2 = \V(\omega)$, then
$\beta$ is a double root of $\omega(\alpha,x_2)$  by \cite[Prop 2.1 \& Thm 
3.1]{GK1996}. 
From the Puiseux theorem (see e.g. \cite[Cor 13.16]{Ei1995}), there
are exactly two Puiseux series $y_1,y_2$ of $\Ccur_2$ at 
$(\alpha,\beta)$. And for $i=1,2$, from \cite[\S 3.2]{Ka2008}, $y_i\in\Kbarpow$,
hence, $\omega(x_1,y_i)\equiv 0$ and
$y_i(\alpha)=\beta$.
Besides, as $(\alpha,\beta)$ is a node, we have $y_1'(\alpha)\neq y_2'(\alpha)$. 
This concludes the proof of assertion \eqref{ass:power1}.

Let $h\in\I(\Ccur) \cap \KK[x_1,x_2,x_3]$. By Euclidean division, there are 
$u,r\in\KK[x_1,x_2]$ and $m\geq 0$ such that
\[
 (\dy{\omega})^m \cdot h = u (\dy{\omega}\cdot x_3 - \rho_3) + r.
\]
Since $\I(\Ccur)\cap\KK[x_1,x_2]=\scal{\omega}$, $\omega$ divides $r$ 
in $\KK[x_1,x_2]$, so that,
\[
 (\dy{\omega}(x_1,y_i))^m \cdot h(x_1,y_i,z_i) \equiv 0,
\]
for $i=1,2$. 
As $\dy{\omega}(x_1,y_i)$ cannot be identically zero - $\KCcurde$ is finite by 
\ref{hyp:inj}, 
$h(x_1,y_i,z_i)\equiv 0$.

Finally, by \ref{hyp:noeth}, $\KK[\Ccur_3]$ is integral over $\KK[\Ccur_2]$,
so that there is $$h_0 \in\I(\Ccur_3) = \I(\Ccur) \cap \KK[x_1,x_2,x_3]$$ monic 
in $x_3$.
From above, for $i=1,2$, $h_0(x_1,y_i,z_i)\equiv 0$ and
$z_i$ is integral over $\Kbarpow$. As $\KKbar$ is an algebraically
closed field of characteristic 0, $\Kbarpow$ is integrally closed 
\cite[Cor 13.15]{Ei1995}.
Thus, as a fraction, $z_i \in \Kbarpow$.
\end{proof}

\begin{proposition}\label{prop:sing-node}
The following assertions are equivalent:
 \begin{enumerate}\vspace*{-0.2em}
  \item $\yy \in \app(\Ccur_2)$;\item $\yy$ is a node of $\Ccur_2$ and
  \begin{equation}\label{eq:ineq}
    (\dyy{\omega}\cdot\dx{\rho_3}-\dxy{\omega}\cdot\dy{\rho_3})(\yy)\neq 0.
  \end{equation}
 \end{enumerate}
\end{proposition}
\begin{proof}
Assume that $\yy=(\alpha,\beta)$ is a node. We first prove that if
\eqref{eq:ineq} holds then, there are two distinct points of $\Ccur$ that
project on $\yy$. By Lemma~\ref{lem:powersol}, there exist $y_1, 
y_2\in\Kbarpow$ such that $y_1'(\alpha)\neq y_2'(\alpha)$ and 
$y_i(\alpha)=\beta$ and $\omega(x_1,y_i)\equiv 0$, for $i=1,2$.
For $i=1,2$ let $z_i = \frac{\rho_3(x_1,y_i)}{
  \dy{\omega}(x_1,y_i)}$. By Lemma~\ref{lem:powersol},
\[
 \dy{\omega}(x_1,y_i)\cdot z_i \equiv \rho_3(x_1,y_i).
\]
Since $z_i\in \Kbarpow$, by derivation and evaluation in 
$x_1=\alpha$,
\begin{equation}\label{eq:homo}
\scalebox{0.97}{$\displaystyle
\big(\dxy{\omega}(\albe)+y_i'(\alpha)\dyy{\omega}(\albe)\big)
z_i(\alpha)
=\dx{\rho_3}(\albe)+y_i'(\alpha)\dy{\rho_3}(\albe).
$}
\end{equation}
By Lemma~\ref{lem:powersol}, $\omega(x_1,y_i)\equiv 0$. 
Differentiating twice and evaluating in $\alpha$, we get
\[
  \dxx{\omega}(\albe) + 2y_i'(\alpha) \dxy{\omega}(\albe) + 
y_i'(\alpha)^2\dyy{\omega}(\albe)=0.
\]
Since $y_1'(\alpha)\neq y'_2(\alpha)$ by Lemma~\ref{lem:powersol}, they
are simple roots of $$\dxx{\omega}(\albe) + 2U\dxy{\omega}(\albe) 
+ U^2\dyy{\omega}(\albe)\in \KKbar[U].$$ Therefore,
\begin{equation}\label{eq:derivnonzero}
\dxy{\omega}(\albe)+y_i'(\alpha)\dyy{\omega}(\albe) \neq 0.
\end{equation}
Now let $H\colon\KKbar\to\KKbar$ such that for all $t\in\KKbar$ 
\[
\scalebox{1}{$\displaystyle
  H(t) = \frac{\dx{\rho_3}(\albe)+t\cdot\dy{\rho_3}(\albe)}{\dxy{\omega}(\albe) 
+ t\cdot\dyy{\omega}(\albe)}.
$}
\]
Using \eqref{eq:homo} and according to \eqref{eq:derivnonzero},  
$H(y_i'(\alpha))=z_i(\alpha)$ for $i=1,2$.
But $H$ is either bijective or constant, whether 
\eqref{eq:ineq} respectively holds or not. As $y_1'(\alpha)\neq
y'_2(\alpha)$, \eqref{eq:ineq} holds if, and only if, $z_1(\alpha)\neq
z_2(\alpha)$. By Lemma~\ref{lem:powersol}, \eqref{ass:power2},
$\zz_1=(\alpha,\beta,z_1(\alpha))$ and $\zz_2=(\alpha,\beta,z_2(\alpha))$ are
points of $\Ccur_3$ projecting on $\yy$. From \ref{hyp:bij}, there are
$\xx,\xx'$ in $\Ccur$ that project on resp. $\zz_1$ and $\zz_2$. They are 
distinct if, and only if, \eqref{eq:ineq} holds.

We can now prove the equivalence statement. We just proved that, if $\yy$ is a
node and \eqref{eq:ineq} holds then, $\yy$ is the projection of two distinct
points, that cannot be singular by \ref{hyp:inj}. Conversely, either $\yy$ is
not a node, and we conclude by \ref{hyp:app} or, by the above discussion, it is
the projection of a point of $\Ccur$, with two distinct tangent lines (that
project on the ones of $\yy$). Hence, $\yy$ is the projection of a singular
point and then, not in $\app(\Ccur_2)$, by definition.
\end{proof}

\section{Connectivity recovery}
We now investigate the connectivity relation between $\CcurR$ and $\CcurRde$. 
The following lemma is partly adapted from \cite[Lemma 6.2]{Ka2008}.
\begin{lemma}\label{lem:realcrit}
  Let $\xx =(\xx_1\dotsc,\xx_n) \in \KCcur$, then $\xx \in \RR^{n}$ if and only 
if $\xx_1 \in \RR$, and
\(
  \KCcurde - \app(\Ccur_2) = \pi_2\big(\KCcur\big).
\)
\end{lemma}
\begin{proof}
The second point is a direct consequence of \ref{hyp:projtang}, 
as the non-singular critical points of $\Ccur$ project to the ones of $\Ccur_2$.

Let $\xx \in \KCcur$, and assume $\xx_1 \in \RR$. By \cite[Prop 3.1]{GK1996}, 
as $\Ccur$ is in 
generic position, computing sub-resultant 
sequences gives a rise to $\sigma_2 \in \KK[x_1]$ such that 
$\xx_2=\sigma_2(\xx_1)\in\RR$.
By \ref{hyp:nosingsecant}, the line $\V(x_1-\xx_1,x_2-\xx_2)$ 
intersects $\Ccur$ at exactly one point. Hence, by \cite[Thm 3.2]{CLO2015}, 
computing a Gr\"{o}bner basis of the ideal $$\I(\Ccur) + \scal{x_1-\xx_1, 
x_2-\xx_2} \subset \RR[\XX]$$ with respect to the lexicographic order 
$x_1\prec\cdots\prec x_n$ gives a rise to $n-2$ polynomials $\sigma_3, \dotsc, 
\sigma_n$ such that $\sigma_i \in \RR[\XXi[i-1]]$ and 
$\sigma_i(\xx_1,\dotsc,\xx_{i-1}) = \xx_i$, for $3\leq i \leq n$.
Hence, the triangular system formed by the $\sigma_i$'s raises polynomials 
$\tau_2, \dotsc, \tau_n \in \RR[x_1]$ such that $\xx_i = \tau_i(\xx_1)$ for 
$i\geq 2$, thus $\xx \in \RR^n$. The converse is straightforward.
\end{proof}

The following lemma shows that, except at apparent singularities, the real 
traces of $\Ccur$ and $\Ccur_2$ share the same connectivity properties.
\begin{lemma}\label{lem:pi2homeo}
The restriction of $\pi_2$ to $\CcurR - \pi_2^{-1}( 
\app(\Ccur_2))$ is a \SA homeomorphism of inverse $\phi_2$, defined 
on $\CcurRde - \app(\Ccur_2)$ \remi{such that}
\[
   \text{for all $\yy \notin \KCcurde$,}\quad
   \scalebox{1}{$\displaystyle
   \phi_2(\yy) = \Big(\yy,\frac{\rho_3(\yy)}{\dy{\omega}(\yy)},
\dotsc,\frac{\rho_n(\yy)}{\dy{\omega}(\yy)}\Big).
    $}
\]
\end{lemma}
\begin{proof}
Consider $\yy \in \CcurRde - \app(\Ccur_2)$. As $\Ccur_2 =
\V(\omega)$, either $\dy{\omega}(\yy)$ is non-zero or $\yy \in \KCcurRde - 
\app(\Ccur_2)$.
In the latter case, according to 
Lemma~\ref{lem:realcrit}
$$\pi_2^{-1}(\yy) \cap \Ccur \subset \KCcurR.$$
By \ref{hyp:nosingsecant} there is a unique 
 $\xx \in \KCcurR - \pi_2^{-1}(\app(\Ccur_2))$ such that $\pi_2(\xx) = \yy$.
Let $\phi_2:\CcurRde - \app(\Ccur_2) \to \RR^n$ be  
defined as:
\begin{itemize}[label=$\rhd$]
\item \textbf{if} $\yy \in 
\KCcurde-\app(\Ccur_2)$, then $\phi_2(\yy)$ is the unique $\xx$ satisfying 
$\pi_2(\xx)= \yy$;
\item \textbf{else} $\phi_2(\yy)=
\big(\yy,\:(\rho_3/\dy{\omega})(\yy), \dotsc,\: (\rho_n/\dy{\omega})
(\yy)\big).$ 
\end{itemize}
Since its graph is a \SA set by construction, $\phi_2$ is a \SA map 
according to \cite[\S 2.5.2]{BPR2016}.
Moreover, if $\yy \in \CcurRde - \app(\Ccur_2)$, then $\phi_2(\yy)$ is the 
unique element of $\CcurR - \pi_2^{-1}(\app(\Ccur_2))$ such that 
$\pi_2(\phi_2(\yy)) = \yy$. 

Since $\dy{\omega}(\yy)$ does not vanish on this set, $\phi_2$ is 
continuous on $\CcurRde - \KCcurde$. We prove that it is continuous 
everywhere.
Let $\yy \in \KCcurRde-\app(\Ccur_2)$ and suppose there is 
a \SA path $\gamma:[0,1] \to \CcurRde$, such that 
$\gamma(0) = \yy$ and $\gamma(t) \in  \CcurRde - 
\KCcurde$, for all $t>0$.
Consider the \SA path $\tau:t\in(0,1]\mapsto \phi_2(\gamma(t)) 
\in \CcurR$. Since $\pi_2$ is a proper map by \ref{hyp:noeth}, $\tau$ is 
bounded. 
Thus, by \cite[Prop 3.21]{BPR2016}, $\tau$ can be continuously 
extended in $t=0$ and by continuity, $\tau(0) \in \CcurR$ and 
$\pi_2(\tau(0))= \pi_2(\phi_2(\yy))=\yy$.
Hence, by uniqueness $\tau(0) = \phi_2(\yy)$ and, by \cite[Prop 3.6 \& 
3.20]{BPR2016}, $\phi_2$ is continuous in $\yy$.
Since $\KCcurde$ is finite, no such path $\gamma$ exists if, and only if, both 
$\yy$ and $\xx$ are isolated points 
so that $\phi_2$ is trivially 
continuous at $\yy$.

In conclusion, $\phi_2$ is a \SA map, continuous on $\CcurRde - \app(\Ccur_2)$, 
of inverse the restriction of $\pi_2$ to $\CcurR - \pi_2^{-1}(\app(\Ccur_2))$ 
by Lemma~\ref{lem:realcrit}.
Hence, this latter restriction is a \SA homeomorphism, as stated.
\end{proof}

It remains to investigate how the connectivity of the real traces of $\Ccur$ and
$\Ccur_2$ are related close to apparent singularities.
Recall that an (ambient) isotopy of $\RR^n$ is a \remi{continuous map} 
$\Hiso\colon \RR^n \times [0,1] \to \RR^n$ such that $\yy\mapsto 
\Hiso(\yy,0)$ is the identity map and $\yy\mapsto 
\Hiso(\yy,t)$ is a homeomorphism for $t \in [0,1]$.
Then two subsets $Y$ and $Z$ of $\RR^n$ are isotopy equivalent
if there is an isotopy $\Hiso$ of $\RR^n$ such that $\Hiso(Y,1)=Z$.

Recall also that a graph $\Gg$ is the data of a set $\Vg$ of vertices, together 
with a set $\Eg$ of edges $\{\bvv,\bvv'\}$, where $\bvv,\bvv'\in\Vg$. 
For any $\yy,\yy' \in \RR^2$, we will denote by $[\yy,\yy']$, the closed line 
segment $\{ (1-t)\yy + t\yy', t\in[0,1] \}$. Then, if $\Vg \subset \RR^2$, we 
call the piecewise linear curve, denoted $\Ccur_{\Gg}$, associated to $\Gg$ the 
union of 
$[\bvv,\bvv']$ for all $\{\bvv,\bvv'\}\in\Eg$.
In the following, we note $\Pset_2 = \pi_2(\Pset)$. 
\begin{definition}\label{def:topology2d}
 Let $\Ggp=(\Vgp,\Egp)$ be a graph, with $\Vgp \subset \RR^2$.
 Then we say that $\Ggp$ is a \emph{real topology graph} of 
$(\Ccur_2,\Pset_2)$ if
 \begin{enumerate}
\item $\CcurRde$ is isotopy equivalent to $\Ccur_{\Ggp}$;
  \item the points of $\KCcurRde\cup \PsetRde$ are embedded in $\Vgp$; 
  \item no two points of $\KCcurRde$ have adjacent vertices in $\Gg$.
 \end{enumerate}
\end{definition}
For the rest of this section, let $\Ggp$ be a \emph{real topology 
graph} of $(\Ccur_2,\Pset_2)$,
$\Hiso$ the induced isotopy and, for $t\in[0,1]$, $\Hiso_t: \yy 
\in \RR^2 \to \Hiso(\yy,t)$, so that $\Hun(\Ccur_{\Ggp}) = \CcurRde$.

Consider \SA paths $\gamma_1,\dotsc,\gamma_4$ in $\RR^2$, 
\review{all starting from a unique point $\pp \in \RR^2$, and not intersecting 
each other elsewhere}
(see Figure~\ref{fig:doublenode}), so that the $\gamma_i$'s can be pairwise 
associated with respect 
to their unique \emph{opposite branch at $\pp$}:
given an orientation of $\RR^2$ and a sufficiently small circle centered at 
$\pp$, we arrange the $\gamma_i's$ around $\pp$ with respect to their 
unique intersection with this circle 
\cite[Thm 9.3.6]{BCR1998};
we then pairwise associate them to the one after next in the 
above arrangement (it does not depend on the chosen 
orientation). Up to reindexing, say that $(\gamma_1,\gamma_3)$ and 
$(\gamma_2,\gamma_4)$ are the \emph{unique couples of opposite branches at 
$\pp$}.

The next lemma follows directly from classical results in knots and braids 
theory, see \cite[Prop 1.9-10]{BZ2003} for the key arguments.
\begin{lemma}\label{lem:arrangebranch}
Let the $\gamma_i$'s as above, and \remi{any} isotopy $\Hisot$ of $\RR^2$.
The curves $(\Hunt(\gamma_1),\Hunt(\gamma_3))$  and 
$(\Hunt(\gamma_2),\Hunt(\gamma_4))$ do not intersect each other, except at 
$\Hunt(\pp)$. They are the unique couples of opposite branches at this point.
\end{lemma}
This property allows us to deduce relations between edges of $\Ggp$, from 
relations between the associated branches of $\CcurRde$.

\begin{figure}[h]\centering
\vspace{-0.4cm}
\begin{minipage}{0.5\linewidth}
 {\setlength{\fboxsep}{0pt}\setlength{\fboxrule}{0pt}\fbox{\includegraphics[width=\linewidth]{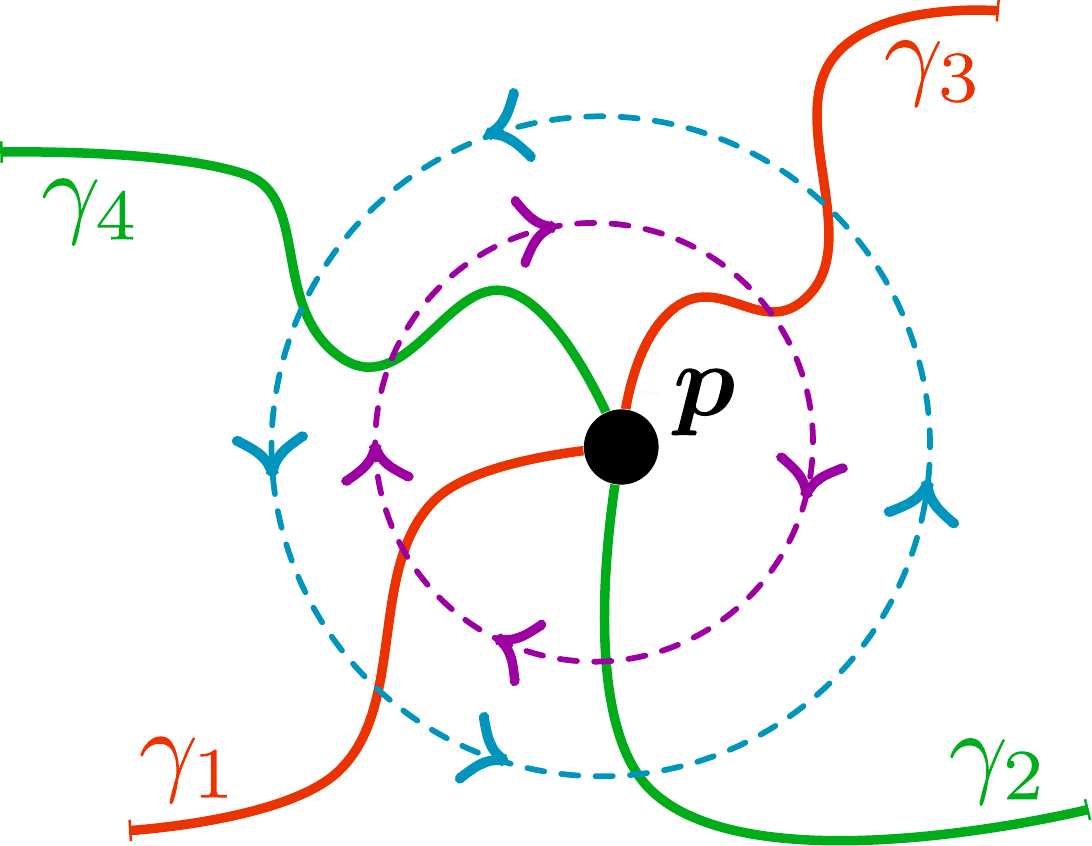}}}\end{minipage}\hspace*{0.3cm}
\begin{minipage}{0.5\linewidth}
 {\setlength{\fboxsep}{0pt}\setlength{\fboxrule}{0pt}\fbox{ \includegraphics[width=0.8\linewidth]{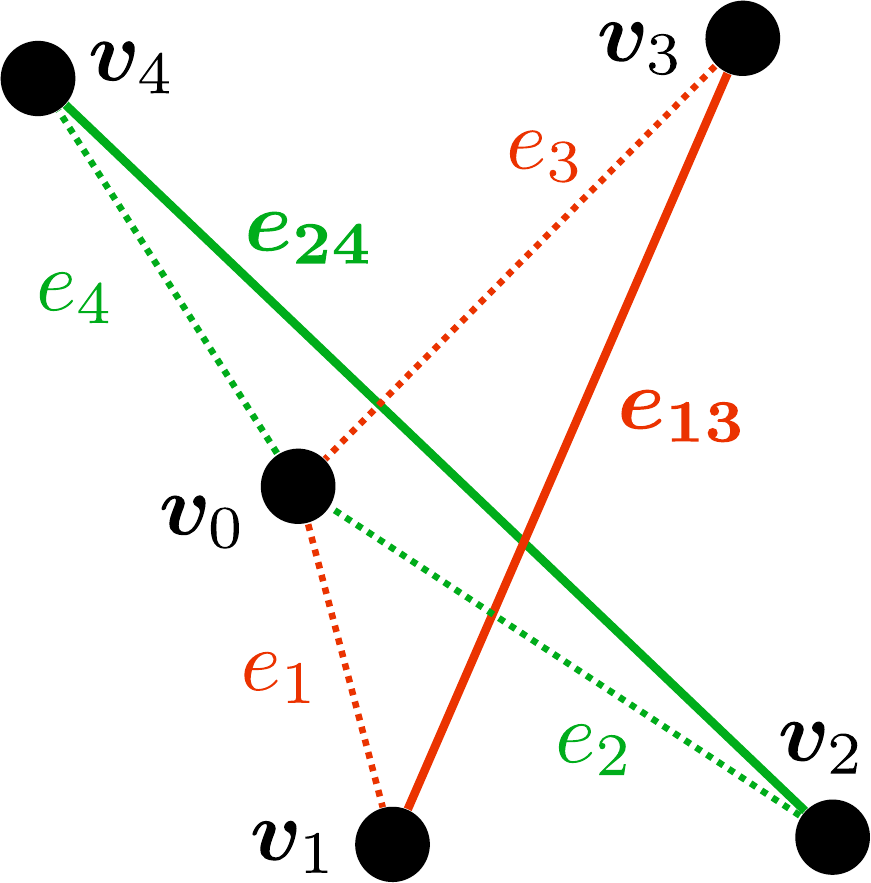}}}\end{minipage}
\caption{\small
The left figure illustrates the context of Lemma~\ref{lem:arrangebranch} 
with two possible ordering of the branches;
the braid structure appears clearly.
On the right, an illustration shows how \NodeRes (Definition~\ref{def:noderes}) 
modifying $\Ggp$ at vertices of $\Vapp$; dotted and solid lines representing 
respective edges of $\Ggp$ and 
$\Gg$.
}
 \label{fig:doublenode}
 \vspace{-0.5cm}
\end{figure}

\begin{lemma}\label{lem:vertexapp}
 Let $\yy=(\alpha,\beta) \in \app(\CcurRde)$. There are exactly five 
distinct vertices ${\bvv}_0, \dotsc, \bvv_4 \in \Vgp$ such that $\Hun(\bvv_0) = 
\yy$ 
and 
for $1\leq i \leq 4$:
\begin{enumerate}
  \item $\{\bvv_0,\bvv_i\}\in\Egp$ and $\Hun(\bvv_i)
  \notin \app(\Ccur_2)$;
  \item if $e_i=[\bvv_0,\bvv_i]$, the $e_i's$ do not cross each other except at 
  $\bvv_0$;
  \item there exists unique \SA paths $\tau_1,\dotsc,\tau_4$ such that for 
  \adrien{\[\tau_i:[0,1]\to \CcurR, \ \ \left\{
  \begin{array}{l}
  \pi_2(\tau_i([0,1])) = \Hun(e_i)\\[0.1em]
  \pi_2(\tau_i(0)) = \yy
  \end{array}\right.
  \]}
  \item assume that $(e_1,e_3)$ and $(e_2,e_4)$ 
are the two unique couples of opposite edges of $\Ggp$ at $\bvv_0$.
Then, there exist $\xx_1 \neq \xx_2$ in $\pi_2^{-1}(\yy)\cap \CcurR$, 
such that $\xx_1=\tau_1(0)=\tau_3(0)$ and $\xx_2=\tau_2(0) = \tau_4(0)$.
 \end{enumerate}
\end{lemma}
\begin{proof}
  Let $\bvv_0=\Hun^{-1}(\yy)$. As $\yy$ is a node, there are exactly four
  distinct vertices $\bvv_1,\dotsc,\bvv_4 \in \Vgp$ such that $\{\bvv_0,\bvv_i\}
  \in \Egp$, for $1\leq i\leq 4$. Indeed, for $1\leq i\leq 4$, let
\[
 e_i:t\in[0,1]\mapsto \bvv_0 +t(\bvv_i-\bvv_0) \in \RR^2
\]
and $\gamma_i=\Hun\circ e_i$. Then the $\gamma_i$'s are the four branches of 
$\CcurRde$ incident in $\yy$.
Remark that, by the third item of Definition~\ref{def:topology2d}, none of the 
$\Hun(\bvv_i)$'s lie in $\KCcurRde$, since $\Hun(\bvv_0) = \yy$ does. 
Besides, by the second item, the $\gamma_i$'s do not intersect $\KCcurde$, 
except in $\yy$.

In particular, the $\gamma_i's$ do not contain points of $\app(\Ccur_2)$ and
intersect each other only at $\yy$. Hence, by Lemma~\ref{lem:arrangebranch}, 
through $\Hun$, the $e_i$'s intersect each other only at $\bvv_0$.

Besides, let $i\in\{1,\dotsc,4\}$, and for $0<t\leq 1$, let $\tau_i(t) =
\phi_2(\gamma_i(t))$, where $\phi_2$ is defined in Lemma~\ref{lem:pi2homeo}.
It is a well-defined \SA path by the above discussion. 
Moreover, by Lemma~\ref{lem:pi2homeo}, 
$\tau_i(t) \in \CcurR$ and $\pi_2(\tau(t)) = \gamma(t) =
\Hun(e_i(t))$, for all $0<t\leq 1$.
Since $\pi_2$ is a proper map by \ref{hyp:noeth}, 
\cite[Prop 3.21]{BPR2016} implies that
$\tau_i$ can be continuously extended in $t=0$. Moreover, by continuity, 
$\pi_2(\tau_i(0))=\yy$.

Finally, $\yy$ being a node, there exist points $\theta_1 \neq 
\theta_2$ in $\RR^2$ and $1\leq i_1,i_2,i_3,i_4\leq 4$ such that,
\[
\theta_1=\gamma_{i_1}'(0)=\gamma_{i_3}'(0)
 \et
\theta_2=\gamma_{i_2}'(0)=\gamma_{i_4}'(0).
\]
This means that the branches $(\gamma_{i_1},\gamma_{i_3})$ and
$(\gamma_{i_2},\gamma_{i_4})$ are the two couples of opposite branches of
$\Ccur_2$ at $\yy$. Then, by Lemma~\ref{lem:arrangebranch}, $(e_{i_1},e_{i_3})$ 
and $(e_{i_2},e_{i_4})$ are the two couples of opposite edges of $\Ggp$ at 
$\yy$. For the sake of clarity assume, without loss of generality that $i_k=k$ 
for all $1\leq k \leq 4$. By continuity, there exist $\vartheta_1 \neq 
\vartheta_2$ in $\RR^n$ such that
\[
 \vartheta_1 = \tau_1'(0)=\tau_3'(0) \et \vartheta_2=\tau_2'(0)=\tau_4'(0),
\]
and $\tau_i(0) \in \pi_2^{-1}(\yy)\cap\CcurR$ for $1\leq i\leq 4$. But as $\yy 
\in\app(\Ccur_2)$, $\pi_2^{-1}(\yy)\cap\Ccur$ contains two distinct 
non-singular points, of distinct tangent lines, \remi{by \ref{hyp:projtang} and 
\ref{hyp:node}}. Since the $\tau_{i}'(0)$'s are tangent lines of $\Ccur$, 
necessarily, $\tau_1(0)$ and $\tau_3(0)$ are equal to one of these points, 
while $\tau_2(0)$ and $\tau_4(0)$ are equal to the other one 
\remi{(if multiple branches converge at a point or the tangent lines 
differ, it becomes singular).}
\end{proof}

If $\Vapp = \Hun^{-1}(\app(\Ccur_2)) \subset \Vgp$ is the subset of apparent
nodes, then Lemma~\ref{lem:vertexapp} provides a procedure to compute a new 
graph $\Gg$, from which we can deduce connectivity queries on $\Ccur$.

\begin{definition}\label{def:noderes}
Let \NodeRes be the procedure that takes as input $\Ggp$ and $\Vapp$ as above 
and outputs the graph $\Gg = (\Vg,\Eg)$ as follows (we keep notations 
of Lemma~\ref{lem:vertexapp}).
\begin{enumerate}[label=\arabic*.]
\item For all $\bvv \in \Vapp$, compute the adjacent vertices 
$\bvv_1,\dotsc,\bvv_4$  of $\bvv$, indexed such that $(e_1, e_3)$ and 
$(e_2,e_4)$ are opposite edges.
\item Remove $\bvv$ from $\Vgp$ and replace the 
four edges $(\{\bvv, \bvv_k\})_{1\leq k\leq 4}$ by the two edges $(\{\bvv_j, 
\bvv_{j+2} \})_{k=1,2}$, as depicted in Figure~\ref{fig:doublenode}.
\end{enumerate}
\end{definition}

We say that $\bvv,\bvv'\in\Vg$ are connected in a graph $\Gg=(\Vg,\Eg)$ if
there exists an ordered sequence $(\bvv_0,\dotsc,\bvv_{N+1})$ of vertices in 
$\Vg$ such that $\bvv_0=\bvv$, $\bvv_{N+1}=\bvv'$ and $\{\bvv_i,\bvv_{i+1}\} 
\in \Eg$, for all $0\leq i \leq N$.
\begin{proposition}\label{prop:algonode}
 Let $\Gg = (\Vg,\Eg)$ be the graph output by \NodeRes, on input $\Ggp$ and 
$\Vapp$. Then,
\begin{enumerate}
 \item $\pi_2(\PsetR) \subset \Hun(\Vg)$;
 \item $\yy,\yy'\in\PsetR$ are \SAC in $\CcurR$ if, and 
only if, $\Hun^{-1}(\pi_2(\yy))$ and $\Hun^{-1}(\pi_2(\yy'))$ are 
connected in $\Gg$.
\end{enumerate}
\end{proposition}
\begin{proof}
\ref{hyp:inj} and \ref{hyp:nosingsecant} 
imply $\pi_2(\Pset) \cap \Hun(\Vapp) = \emptyset$. 
Then $\PsetRde = \pi_2(\PsetR)$ as $\pi_2$ is injective on $\Pset$, and, by 
definition, $\PsetRde \subset \Vg$.

We now deal with the second statement. Let $\xx,\xx' \in \PsetR$ and
\[
\bvv = \Hun^{-1}(\pi_2(\xx)) \et \bvv' = \Hun^{-1}(\pi_2(\xx'))
\] in $\Vg$.
Assume first that $\bvv$ and $\bvv'$ are connected in $\Gg$. Then there 
exist $\bvv_1,\dotsc,\bvv_N \in \Vg$ such that, if $\bvv_0=\bvv$ and 
$\bvv_{N+1}=\bvv'$, then $\{\bvv_i,\bvv_{i+1}\} \in \Eg$ and $\Hun(\bvv_i) 
\notin \app(\Ccur_2)$ for $0\leq i \leq N+1$.
Fix $i \in \{0,\dotsc,N\}$. 
By Lemma~\ref{lem:pi2homeo}, $\xx_i=\phi_2(\Hun(\bvv_i))$ and 
$\xx_{i+1}=\phi_2(\Hun(\bvv_{i+1}))$ are well-defined in $\CcurR$ 

If $\{\bvv_i,\bvv_{i+1}\} \in \Egp$ then, $\Hun([\bvv_i,\bvv_{i+1}]) \cap
\app(\Ccur_2) = \emptyset$, and, by Lemma~\ref{lem:pi2homeo}, $\xx_i$ and
$\xx_{i+1}$ are \SAC in $\CcurR$ through $\phi_2$. 
Otherwise, $\{\bvv_i,\bvv_{i+1}\} \notin \Egp$, and, by construction of $\Gg$, 
there exists $\ww \in \Vapp$ such that $\{\bvv_i,\ww\}$ and 
$\{\ww,\bvv_{i+1}\}$ are in $\Egp$. However, since $\{\bvv_i,\bvv_{i+1}\} \in 
\Eg$, then, according to the construction of $\Gg$, 
$$e_i=[\ww,\bvv_i] \et e_{i+1}=[\ww,\bvv_{i+1}]$$
are opposite edges of $\Ggp$ at $\ww$. Hence, by 
items \emph{(2)} and \emph{(3)} of Lemma~\ref{lem:vertexapp}, there exists a 
\SA path 
$\tau:[-1,1]\to\CcurR$ connecting $\xx_i$ to $\xx_{i+1}$. All in all, by 
transitivity, $\xx_0=\xx$ and $\xx_{N+1}=\xx'$ are \SAC in $\CcurR$, and we are 
done.

Conversely, suppose that $\xx$ and $\xx'$ are \SAC in $\CcurR$ and let
$\tau:[0,1]\to\CcurR$ be a \SA path such that $\tau(0)=\xx$ and $\tau(1)=\xx'$.
Let $\gamma = \pi_2 \circ \tau$, and 
\[
\{t_1,\dotsc,t_N\} = \gamma^{-1}\big(\Hun(\Vgp)\big) \subset (0,1)
\]
such that $t_1 < \dotsc < t_N$.
Let $t_0 = 0$, $t_{N+1}=1$ and for $0\leq i\leq N+1$, $\bvv_i = 
\Hun^{-1}(\gamma(t_i)) \in \Vgp$. By assumption, $\{\bvv_i,\bvv_{i+1}\} \in 
\Egp$ for all $i\in\{0,\dotsc,N\}$.
Let us prove by induction that for $0\leq i\leq N+1$, either $\bvv_i 
\in \Vapp$ or $\bvv_i$ is connected to $\bvv_0$ in $\Gg$. 
If $i=0$, there is nothing to prove, so let $1\leq i\leq N$ and suppose that 
the statement holds for all $0\leq j\adrien{<} i$. 

Assume $\bvv_{i+1} \notin \Vapp$. Then, either $\bvv_i \notin \Vapp$, and, 
by induction hypothesis, $\bvv_{i+1}$ and $\bvv_0$ are connected, through 
$\bvv_i$, in $\Gg$.
Either $\bvv_i \in \Vapp$ and, by Lemma~\ref{lem:vertexapp}, there
are exactly four distinct $\ww_1,\ww_2,\ww_3,\ww_4 \in \Vg-\Vapp$ such that
$\{\bvv_i,\ww_j\} \in \Egp$, for $1\leq j \leq 4$. Assume, without loss of
generality, that $\bvv_{i+1} = \ww_1$. Then, there is $j_1\in\{2,3,4\}$ such
that $\bvv_{i-1}=\ww_{j_1}$. Using the notation of Lemma~\ref{lem:vertexapp},
assume, without loss of generality, that $e_3=[\bvv_i,\ww_3]$ is the opposite
branch of $e_1=[\bvv_i,\ww_1]$ in $\Ggp$ at $\bvv_i$. Then, by items \emph{(2)} 
and \emph{(3)} of Lemma~\ref{lem:vertexapp}, we have 
$j_0=3$, since $\tau([t_{i-1},t_i])$ is connected to $\tau([t_i,t_{i+1}])$. By 
construction of $\Gg$, $\ww_1=\bvv_{i+1}$ is connected to $\ww_3=\bvv_{i-1}$ in 
$\Gg$, so that, by induction, $\bvv_{i+1}$ is connected to 
$\bvv_0$, through $\bvv_{i-1}$. Hence, $\bvv = \bvv_{N+1}$ and
$\bvv'=\bvv_{0}$ are connected in $\Gg$, proving the converse.
\end{proof}\vspace*{-0.3em}
\review{Proposition~\ref{prop:algonode} also implies that $\Gg$ and $\CcurR$ 
share 
the same number of \SACCs. Therefore, by computing $\Gg$, one can determine 
this 
number and answer connectivity queries on $\PsetR$.}

\section{Algorithm}
We now provide an algorithm for solving connectivity queries over real algebraic
curves, whose different steps correspond sequentially, except for one, to the
different sections of this document.

Given a sequence of polynomials defining an algebraic curve, the first step
is to perform a linear change of variable, generic enough to ensure assumption
$\Hspace{}$, and to compute a \ODP encoding it. Answering connectivity queries 
on
the sheared curve is equivalent to do so on the original curve. By \cite[Thm
6.18]{GM2019} (or \cite[Prop 6.3]{SS2017}), computing such a
parametrization has complexity cubic in the degree of the curve, thus bounded
by our overall complexity. Besides, according to \cite[\S\,J]{SS2017}, changing 
variables in zero and \ODPs has similar complexity. 
Hence, for the sake of clarity, we omit these two steps.

Following the state of the art of curve topology computation, we consider
polynomials with integer coefficients, so that $\KK=\QQo$, $\RR=\RRo$ and
$\KKbar=\CCo$. Moreover, we denote by $\infx$ the preorder on points of 
$\RRo^n$ w.r.t. the first coordinate, when they are distinct.

\subsection{Subroutines}
We assume that $\Rparam=(\omega,\rho_3,\dotsc,\rho_n)$
has coefficients in $\ZZo$ and magnitude $(\delta,\tau)$,
and consider a \ZDP $\Pparam = (\lambda, \vtheta_2, \dotsc, 
\vtheta_n)$, with coefficients in $\ZZo$ and magnitude $(\deltaP,\tauP)$ 
encoding $\Pset$. 
Note that $\Dparam = (\omega,\rho_2)$ and $\Qparam = (\lambda,\vtheta_2)$ 
are parametrizations encoding respectively $\Ccur_2$ and $\Pset_2$. We denote 
further $R = \Res_{x_2}(\omega,\dy{\omega})$. Since, by \ref{hyp:odp}, $\omega$ 
is monic in $x_2$, its roots are exactly the abscissas of $\KCcurde$.
From \ref{hyp:inj}, points of $\app(\Ccur_2)$ can be identified by their 
abscissa, \adrien{which}, following Proposition~\ref{prop:sing-node}, can be 
reduced to 
gcd computations.
\begin{proposition}\label{prop:procappsing}
 There exists an algorithm \AppSing taking as input $\Rparam$, as above, and 
computing a square-free polynomial $\qapp \in \ZZo[x_1]$, of magnitude 
$(\delta^2, \Otilde(\delta^2+\delta\tau))$ such that
\[
 \app(\Ccur_2) = \{ (\alpha,\beta) \in \KCcurde \mid \qapp(\alpha)=0\},
\]
using $\Otilde(\delta^6+\delta^5\tau)$ bit operations.
\end{proposition}
\begin{proof}
 Let $(\alpha,\beta)\in\KCcurde$.
 According to \cite[Thm 3.2.(ii)]{Ka2008}, since $\Ccur$ satisfies 
$\Hspace{}$, $(\alpha,\beta)$ is a node if, and only if, $\alpha$ is a double 
root of $R$, i.e. if, and only if, $\alpha$ 
is a root of
\[
 q = \gcd(R^*,R') / \gcd(R^*,R',R''),
\]
where $R^*$ is the square-free part of $R$.
Moreover,  let $(\sr_1,\sr_{1,0})$ be the first subresultant sequence of 
$(\omega,\dy{\omega})$. By \cite[Thm 3.1]{GK1996}, if $q(\alpha)=0$ then, 
$\sr_1(\alpha) \neq 0$, and \[
 \sr_1(\alpha)\cdot \beta = -\sr_{1,0}(\alpha).
\]

Let $A(x_1,x_2)$ be the polynomial on the left-hand side of \eqref{eq:ineq} in 
Proposition~\ref{prop:sing-node}, and $u$ be a new indeterminate.
Let $\Ahom(x_1,x_2,u)$ be the homogenization of $A$ in $x_2$, and 
$B=\Ahom(x_1,-\sr_{1,0}, \sr_1)$.
Then, from Proposition~\ref{prop:sing-node}, the square-free polynomial 
\[
 \qapp=q/\gcd(q,B)
\] vanishes at $\alpha$ if, and only if, $(\alpha,\beta)\in\app(\Ccur_2)$, as 
required.

We now deal with the quantitative bounds.
By \cite[Lemma 14]{MSW2015}, $R$, $R^*$, $\sr_1$ and 
$\sr_{1,0}$ have magnitude $(\delta^2,\Otilde(\delta^2+\delta\tau))$ and can be 
computed using $\Otilde(\delta^6+\delta^5\tau)$ bit operations.
Hence, by \cite[Cor 11.14]{VG2013} and \cite[Lemma 12]{MSW2015}, 
computing $\gcd(R^*,R')$, $\gcd(R^*,R',R'')$ and then $q$ can be done using 
$\Otilde(\delta^4+\delta^3\tau)$ bit operations.
Moreover, by \cite[Lemma 11]{MSW2015}, $q$ has magnitude $(\delta^2, 
\Otilde(\delta^2+\delta \tau))$.

Besides, $\Ahom$ has magnitude $(O(\delta),\Otilde(\tau))$, so that 
$B$ has magnitude $$\big(\Otilde(\delta^3), 
\Otilde(\delta^3+\delta^2\tau)\big).$$
Hence, by \cite[Cor 11.14]{VG2013} computing, $\gcd(q,B)$ 
requires $\Otilde(\delta^{6} + \delta^{5}\tau)$ bit operations. 
From this, computing $\qapp$ costs $\Otilde(\delta^4+\delta^3\tau)$ bit 
operations, by \cite[Prop 2.15]{DDRRS2022}.
Finally, $\qapp$ has magnitude $(\delta^2, \Otilde(\delta^2+\delta 
\tau))$, by \cite[Lemma 11]{MSW2015}.
\end{proof}
\vspace*{-0.2em}
Suppose now that the polynomial $\qapp$, from 
Proposition~\ref{prop:procappsing}, has been computed. 
We can compute a real topology graph of $(\Ccur_2,\Pset_2)$, while identifying 
the vertices corresponding to $\app(\Ccur_2)$ and $\Pset_2$.
\begin{proposition}\label{prop:topoplane}
There exists an algorithm \TopoPlane taking as input $\Rparam$, $\Qparam$ 
and $\qapp$ as above and computing $\Gg = (\Vg,\Eg)$, a real topology graph of 
$(\Ccur_2, \Pset_2)$, of size at most $O(\delta^3+\delta\deltaP)$, using 
\[
\Otilde(\delta^6 + \delta^5\tau +  \deltaP^6 + \deltaP^5\tauP)
\]
bit operations.
It also outputs ordered sequences $\Vseqapp$ and $\VseqPparam$, of elements of 
$\Vg$, that are in one-to-one correspondence with resp. the 
points of $\app(\CcurRode)$ and $\PsetRode$, ordered with respect to $\infx$.
\end{proposition}
\begin{proof}
According to \cite[Thm 14]{KS2015}, and more recently \cite[Thm 
1.1]{DDRRS2022}, there is an algorithm that computes a planar graph $\Gg$, 
whose associated piecewise linear curve $\Ccur_{\Gg}$, is isotopy equivalent 
to 
$\CcurRode$, using $\Otilde(\delta^6 + \delta^5\tau)$ bit operations.
Under slight modifications, these algorithms can compute the 
claimed output of $\TopoPlane$, within the same complexity bounds.
For clarity, we only consider the algorithm of \cite{DDRRS2022}, 
that we roughly describe.

Let $\alpha_1<\cdots<\alpha_N$ be the abscissas of the points of $\KCcurRode$.
They are distinct by \ref{hyp:inj}.
\cite[Prop 2.24]{DDRRS2022} first computes disjoint isolating intervals for 
each 
$\alpha_i$.
Then, \cite[Prop 3.13]{DDRRS2022} isolates the ordinates of the points above 
each $\alpha_i$. 
This process gives rise to isolating boxes, which stand for vertices in the 
final graph.
The algorithm eventually connects these boxes to separating vertices above 
regular values in the intervals $(\alpha_j,\alpha_{j+1})$. The latter is 
done by counting the number of incoming left and right branches in each box.
For points of $\KCcurRode$, it is tackled by \cite[\S 4.2-4]{DDRRS2022}, 
while for others it is straightforward (exactly one branch 
from each side).

The above process computes a graph $\Gg=(\Vg,\Eg)$, such that $\Ccur_{\Gg}$ is 
isotopy equivalent to $\CcurRode$.
Remark that $\Vg$ contains a subset $\Vcrit$ of vertices associated to the 
unique point of $\KCcurRode$ above the $\alpha_i$'s, all separated by vertices 
associated to regular points.
Moreover, by Proposition~\ref{prop:procappsing}, $\Vapp$ is exactly the subset 
of $\Vcrit$, associated to the $\alpha_i$'s where $\qapp$ vanishes.   
Then, according to \cite[Prop 2.24]{DDRRS2022} and 
Proposition~\ref{prop:procappsing}, one can compute disjoint isolating 
intervals of the roots of $R$ and $\qapp$ and identify all common roots, using
\[
 \Otilde(\delta^6 + \delta^5\tau)\]
bit operations. This gives $\Vseqapp$.

Hence, it remains to show that introducing vertices for control points 
$\PsetRode$ (together with those above and below) can be done in the claimed 
bound.
First, recall that $\Dcal=(\lambda, \vtheta_2)$ encodes $\Pset_2$. According to 
\cite[Prop 2.24]{DDRRS2022} again, we can compute disjoint isolating 
intervals for all distinct (by \ref{hyp:inj}) real roots of $\lambda$ and $R$, 
using at most
\[
 \Otilde(\delta^6+ \delta^5\tau + \deltaP^6 + \deltaP^5\tauP)\]
bit operations.
Next, let $g(x_1,x_2) = \lambda'\cdot x_2-\vtheta_2$. It is a bivariate 
polynomial with magnitude $(\deltaP,\tauP)$.
Then, according to \cite[Prop 3.14]{DDRRS2022}, for each root $\beta$ of 
$\lambda$, we can compute isolating intervals for all roots $\xx_2$ of 
$(\omega\cdot g)(\beta,\xx_2)$, and identify the unique common roots, within 
the same complexity bound. This gives $\VseqPparam$.
Moreover, since $\Pset\cap\KCcurRode = \emptyset$, as seen above, the 
connection step for the introduced vertices is straightforward, and does not 
affect the complexity bound.

Finally, since we consider at most $\delta^2+\deltaP$ fibers, each of them 
containing at most $\delta$ points then, taking in account the regular 
separating fibers, we get at most $O(\delta^3+\delta\deltaP)$ vertices and 
edges.
\end{proof}\vspace*{-1em}

\subsection{The algorithm}
Let \IndexConComp be an algorithm taking as input a graph $\Gg=(\Vg,\Eg)$, and 
an ordered sequence $\Vseq=(\bvv_1,\dotsc,\bvv_N)$ of vertices of $\Gg$. 
It outputs a partition $I_1,\dotsc,I_s$ of $\{1,\dotsc,N\}$, grouping the 
indices of the $\bvv_i$'s lying in the same connected components of $\Gg$.
By \cite[\S 22.2]{CLRS2009}, this has a bit complexity linear 
in the size of $\Gg$.
\begin{algorithm}[H]
\caption{\ConnectCurve}\label{alg:connectcurve}
 \begin{algorithmic}[1] \Require $\Cparam = (\omega,\rho_3,\dotsc,\rho_n) \subset \ZZo[x_1,x_2]$ 
encoding an algebraic curve $\Ccur \subset \CCo^n$ and $\Pparam = 
(\lambda,\vartheta_2,\dotsc,\vartheta_n) \subset \ZZo[x_1]$ encoding points 
$\pp_1\infx\cdots\infx\pp_{\deltaP}$ of $\CcurRo$, such that $(\Ccur,\Pset)$ 
satisfies $\Hspace{}$.
  \Ensure a partition of $\{1,\dotsc,\deltaP\}$ grouping the indices of the 
$\pp_i$'s lying in the same \SACC of $\CcurRo$.
  \State $\Qparam \gets (\lambda,\vartheta_2);$
  \State $\qapp \gets \AppSing(\Rparam)$
  \State $\left[\Ggp,\; \Vseqapp,\: \VseqPparam\right] \gets 
\TopoPlane\left(\Rparam,\;\Qparam,\;\qapp\right);$
  \State $\Gg \gets \NodeRes(\Ggp,\,\Vseqapp);$
  \State\Return $\IndexConComp(\VseqPparam,\,\Gg)$;
\end{algorithmic}
\end{algorithm}

Correction, and complexity estimate, of Algorithm~\ref{alg:connectcurve}, 
follow directly from
Propositions~\ref{prop:procappsing}, 
\ref{prop:topoplane} and \ref{prop:algonode}.
This proves Theorem~\ref{thm:mainresult}.

\review{As mentioned before,}
the number of connected components of the graph $\Gg$ computed 
equals the number of \SACCs of $\CcurRo$.
\remi{
As an extension, for curves given as unions, Algorithm~\ref{alg:connectcurve} 
can be applied to each curve, where query points are extended to 
include pairwise common intersection points. The resulting subsets are then 
merged based on their shared points.}

\balance
\bibliography{biblio}
\bibliographystyle{ACM-Reference-Format}

\end{document}